\documentclass[reqno]{amsart}

\usepackage[a4paper]{geometry}

\usepackage{hyperref}

\usepackage[shortlabels]{enumitem}

\usepackage[colorinlistoftodos]{todonotes}

\usepackage{amssymb, amsbsy, amsfonts}
\usepackage{amsmath, amsthm}
\usepackage{latexsym,float,color}

\allowdisplaybreaks[4]

\def\<#1>{\langle#1\rangle}
\let\set\mathbb
\def\vect#1{\mathbf{#1}}

\def\lcm{\operatorname{lcm}}
\def\coeff{\operatorname{coeff}}

\def\ord{\operatorname{ord}}

\def\den{\operatorname{den}}

\def\rE{$R$}
\def\rpiE{$R\Pi$}
\def\piE{$\Pi$}
\def\pisiE{$\Pi\Sigma$}
\def\sigmaE{$\Sigma$}

\def\KK{\set K}
\newcommand{\NN}{\set N_{\geq0}}
\newcommand{\NP}{\set N_{\geq1}}
\def\ZZ{\set Z}
\def\GG{\set G}
\def\HH{\set H}
\def\AR{\set A}
\def\SA{\set S}
\def\FF{\set F}

\def\EE{\set E}
\def\QQ{\set Q}

\newdimen\listablecorrection
\newtheorem{theorem}{Theorem}
\newtheorem{proposition}{Proposition}
\newtheorem{corollary}{Corollary}
\newtheorem{lemma}{Lemma}
\newtheorem{definition}{Definition}
\newtheorem{example}{Example}

\newtheorem{alg}{Algorithm}

\newcommand{\shift}{\quad}

\newcommand{\myframeV}[2]{\begin{minipage}[t]{#1}#2\end{minipage}}

\newcommand{\fct}[3]{#1:#2\to#3}

\newcommand{\dfield}[2]{(#1,#2)}

\newcommand{\const}[2]{\text{const}_{#2}#1}

\newcommand{\Ann}{\text{Ann}}

\newcommand{\lr}[1]{\langle#1\rangle}

\newcommand{\rpisiE}{$\mathrm{R}\mathrm{\Pi\Sigma}$}

\begin{document}

\vspace*{-1cm}

\begin{flushright}
	RISC Report number 23-01\\[1cm] 
	%{\tt arXiv:XXX [hep-th]}\\[1cm]
\end{flushright}

\author{Carsten Schneider}
\address{Johannes Kepler University Linz\\
	Research Institute for Symbolic Computation\\
	A-4040 Linz, Austria}
\email{Carsten.Schneider@risc.jku.at}
\thanks{Supported by the Austrian Science Foundation (FWF) grant P33530.}

\title[Refined telescoping algorithms to reduce the degrees of the denominators]{Refined telescoping algorithms in \rpisiE-extensions to reduce the degrees of the denominators}

%%
%% The abstract is a short summary of the work to be presented in the
%% article.
\begin{abstract}
We present a general framework in the setting of difference ring extensions that enables one to find improved representations of indefinite nested sums such that the arising denominators within the summands have reduced degrees. The underlying (parameterized) telescoping algorithms can be executed in \rpisiE-ring extensions that are built over general \pisiE-fields. An important application of this toolbox is the simplification of d'Alembertian and Liouvillian solutions coming from recurrence relations where the denominators of the arising sums do not factor nicely. \end{abstract}


\maketitle

\section{Introduction}

Parameterized telescoping, a central paradigm of symbolic summation, can be introduced in a a \emph{difference ring} (or \emph{field}) $\dfield{\AR}{\sigma}$ as follows. $\AR$ is a ring (or field) in which the summation objects are modeled, $\fct{\sigma}{\AR}{\AR}$ is a ring (or field) automorphism that scopes the shift operator, and $\KK=\const{\AR}{\sigma}=\{c\in\AR\mid \sigma(c)=c\}$ is the set of constants which forms a subring (or subfield) of $\AR$; here $\KK$ is always a field also called \emph{constant field}. Then we are interested in the following problem.

\medskip

\noindent\textbf{Problem PT in $\dfield{\AR}{\sigma}$} (with constant field $\KK=\const{\AR}{\sigma}$). \textit{Given} $f_1,\dots,f_m\in\AR\setminus\{0\}$. \textit{Find} $h\in\AR$ and $(c_1,\dots,c_m)\in\KK^m\setminus\{\vect{0}\}$  with
\begin{equation}\label{Equ:ParaT}
\sigma(h)-h=c_1\,f_1+\dots+c_d\,f_m.
\end{equation}
For the special case $m=1$ this reduces to the telescoping problem.

\medskip

\noindent\textbf{Problem T in $\dfield{\AR}{\sigma}$}. \textit{Given} $f\in\AR\setminus\{0\}$. Find $h\in\AR$ with
\begin{equation}\label{Equ:Tele}
	\sigma(h)-h=f.
\end{equation}

If $\sigma$ encodes the shift in $k$, equation~\eqref{Equ:Tele} turns to $h(k+1)-h(k)=f(k)$. Summing this equation over $k$ from $a$ to $b$ gives $\sum_{k=a}^bf(k)=h(b+1)-h(a)$.
Similarly, Zeilberger's creative telescoping paradigm~\cite{Zeilberger:91} for finding recurrences of definite sums is covered in Problem~PT by setting $f_i=F(n+i-1,k)\in\AR$ for $1\leq i\leq m$. 
%Then plugging the found solution of Problem~PT into~\eqref{Equ:ParaT} and summing it over $k$ from $a$ to $b$ yields a recurrence relation in $n$.
%
%\vspace*{-0.5cm}
%
%$$c_1(n)\text{\footnotesize$\sum_{k=a}^b$}F(n,k)+\dots+c_m(n)\text{\footnotesize$\sum_{k=a}^b$} F(n+m-1,k)=g(n,b)-g(n,a).$$
%
%\vspace*{-0.15cm}

The breakthrough of these summation techniques was Gosper's telescoping algorithm for hypergeometric products~\cite{Gosper:78} and Zeilberger's extension to creative telescoping~\cite{Zeilberger:91}. They have been optimized and extended further to other input classes, such as \hbox{($q$--)}hyper\-geometric products~\cite{Paule:95,PauleRiese:97,Bauer:99,CK:12,CJKS:13,ChenHKL15} or holonomic sequences~\cite{Zeilberger:90a,Chyzak:00,Koutschan:13}. Another milestone was Karr's summation algorithm~\cite{Karr:81,Karr:85} that solves Problems~T and~PT in \pisiE-fields.

\vspace*{-0.1cm}

\begin{definition}\label{Def:PSFieldExt}
\normalfont
	A difference field extension $(\FF,\sigma)$ of a difference field $\dfield{\GG}{\sigma}$ is called a \emph{\pisiE-field extension} if 
	$\GG=\GG_0\leq\GG_1\leq\dots\leq\GG_e=\FF$
	is a tower of rational function field extensions with $\GG_i=\GG_{i-1}(t_i)$ for $1\leq i\leq e$ and we have $\const{\FF}{\sigma}=\const{\GG}{\sigma}$ where for all $1\leq i\leq
	e$ one of the following holds:
\newpage
	\begin{itemize}[topsep=0pt, partopsep=0pt, leftmargin=7pt, itemindent=5pt]
		\item$\frac{\sigma(t_i)}{t_i}\in(\GG_{i-1})^*$ ($t_i$ is called a \emph{\piE-field monomial});
		\item $\sigma(t_i)-t_i\in\GG_{i-1}$ ($t_i$ is called a \emph{\sigmaE-field monomial}).
	\end{itemize}
	%Depending on the occurrences of the \pisiE-monomials such an extension is also called \emph{\piE-/\sigmaE-field extension}.
	Such an $\dfield{\FF}{\sigma}$ is called a \emph{\pisiE-field over $\KK$} if $\const{\GG}{\sigma}=\GG=\KK$.
\end{definition}

Together with \rpisiE-extensions~\cite{DR1,DR3} (see Definition~\ref{Def:APSExt}) one can rephrase indefinite nested sums defined over nested products fully automatically~\cite{Schneider:05c,Petkov:10,DR2,OS:18,OS:20,SchneiderProd:20}; see also~\cite{Petkov:10,ZimingLi:11}. 
In particular, improved algorithms for (parameterized) telescoping~\cite{Schneider:08c,Schneider:10a,Schneider:10b,Petkov:10,Schneider:15} are implemented within the summation package~\texttt{Sigma}\cite{Schneider:07a,Schneider:21} to find representations with minimal nesting depth. 
Further important simplifications have been introduced in~\cite{Abramov:75,Paule:95} for the rational case $\KK(x)$ with $\sigma(x)=x+1$ that finds for a given $f\in\KK(x)$ an $h$ in $\KK(x)$ or in a \pisiE-field $\KK(x)(t)$ with $\sigma(t)-t=f'\in\KK(x)$ such that~\eqref{Equ:Tele} holds and the denominator of $f'$ has minimal degree; for the generalization in a \pisiE-field $\dfield{\FF(x)}{\sigma}$ we refer to~\cite{Schneider:07d}.

In this article we aim at enhancing this telescoping approach~\cite{Abramov:75,Paule:95,Schneider:07d} (also related, e.g., to~\cite{Abramov:03,BChen:05,CSFFL:15}) such that the generator $x$ may arise also within an extension tower. E.g., consider the sum in
\vspace*{-0.5cm}

\begin{multline}\label{Equ:HarmonicSumId}
	\text{$\sum_{k=1}^n$}\Big(\frac{-2+k}{10 (1+k^2)}+\frac{(1-4 k-2 k^2)}{10 (1+k^2)(2+2 k+k^2)}S_1(k)+\frac{(1-4 k-2 k^2)}{5 (1+k^2)(2+2 k+k^2)}S_3(k)\Big)\\[-0.2cm]
	=\frac{(n^2+4 n+5)}{10(n^2+2 n+2)}S_1(n)-\frac{(n-1) (n+1)}{5 (n^2+2 n+2)}S_3(n)-\frac{2}{5}S_2(n).
\end{multline}
where the denominators in $k$ do not factorize nicely over $\QQ$; here $S_o(n)=\sum_{i=1}^k1/i^o$ denotes the harmonic numbers. Then with our new algorithms 
one can compute the right-hand side of~\eqref{Equ:HarmonicSumId} in terms of sums whose denominators factor linearly.
In general, we assume that the sums and products within $\AR$ have nice denominators in $x$ (here in $k$), i.e., have irreducible factors whose degrees are at most $d$ for some given $d\in\NN$. Then we can decide algorithmically if Problems~T and~PT are solvable in $\AR$ or in an extension of it of where the additional sums have again nice denominators.

These algorithms play a crucial role to simplify d'Alem\-bertian and Liouvillian solutions~\cite{Petkov:92,vanHoeij:99,APP:98,Singer:99,Petkov:2013} for hypergeometric products, and their generalizations in \pisiE-fields~\cite{ABPS:20}. E.g., during calculations coming from particle physics~\cite{BKKS:09,BMSS:22a,BMSS:22b} we have obtained sum solutions up to nesting depth $40$ where the denominators of the sums are built by irreducible polynomials with degrees up to 1000. Using our new toolbox we  have obtained optimal sum representations with only linear factors in the denominators. These simplifications are essential to get solutions in terms of harmonic sums and their generalizations~\cite{Bluemlein:99,Vermaseren:99,ABS:11,ABS:13,ABRS:14}. In particular, these tools can be combined efficiently with quasi-shuffle algebras~\cite{Bluemlein:04,AS:18}.

The article proceeds as follows. In Sec.~\ref{Sec:BasicNotions} we present refined \rpisiE-extensions and their main properties. In Sec.~\ref{Sec:RefinedRep} we elaborate on denominator reduced representations. This insight yields new telescoping algorithms in Sec.~\ref{Sec:RefinedTele}. A conclusion is given in Sec.~\ref{Sec:Conclusion}.

\section{Basic notions and properties}\label{Sec:BasicNotions}

All fields and rings have characteristic $0$.  $\dfield{\EE}{\sigma}$ is a \emph{difference ring (or field) extension} of $\dfield{\AR}{\sigma'}$ if $\AR$ is a subring (or subfield) of $\EE$ and $\sigma|_{\AR}=\sigma'$; from now on we do not distinguish between $\sigma$ and $\sigma'$. 

We call a difference field or ring $\dfield{\AR}{\sigma}$ with constant field $\KK$ \emph{computable} if $\sigma$ is computable, one can carry out the standard operations in $\AR$ and can decide if an element is $0$. It is called \emph{LA-computable} if, in addition, one can compute for $f_1,\dots,f_m\in\AR$ a basis of the $\KK$-vector space 
$$\Ann_{\KK}(f_1,\dots,f_m)=\{(c_1,\dots,c_m)\in\KK^m\mid c_1\,f_1+\dots+c_m\,f_m=0\}.$$
In a \pisiE-field extension $\dfield{\FF(x)}{\sigma}$ of $\dfield{\FF}{\sigma}$ we define the \emph{period} of $h\in\FF^*$ by $\text{per}(h)=0$ if there is no $n\in\NP$ with $\sigma^n(h)/h\in\FF$; otherwise, $\text{per}(h)$ is the smallest $n\in\NP$ with this property.
We rely on the following properties proved for a \pisiE-field in~\cite{Karr:81} and for a \pisiE-field extension in~\cite{Bron:00,Schneider:01}.

\begin{lemma}\label{Lemma:ShiftProp}
	Let $f,g\in\FF[x]\setminus\{0\}$ in a \pisiE-extension $\dfield{\FF(x)}{\sigma}$ of $\dfield{\FF}{\sigma}$. 
\begin{enumerate}[topsep=0pt, partopsep=0pt, leftmargin=7pt, itemindent=5pt,label={\arabic*.}]
		\item If $\text{per}(f)>0$, then $\frac{\sigma(x)}{x}\in\FF$ and $f=c\,x^m$ with $c\in\FF^*$, $m\in\ZZ$.
		\item Suppose that $\frac{\sigma(x)}{x}\notin\FF$ or not both $f,g$ have the form $c\,x^m$ with $c\in\FF^*$, $m\in\ZZ$. Then there is at most one $k\in\ZZ$ with $\sigma^{k}(f)/g\in\FF$.
	\end{enumerate}
\end{lemma}

\noindent Thus any element in $\FF(x)$ has period $0$ or $1$. Furthermore, the only monic and irreducible polynomial with period $1$ is the \piE-monomial $x$ itself. Write $f=f_1^{n_1}\dots f_u^{n_u}\in\FF(x)$ where the irreducible polynomials $f_i$ are pairwise coprime and $n_i\in\ZZ$. We say that \emph{$f$ has $x$-degree{$\leq d$}} with $d\in\NN$ if for any period $0$ factor $f_i$ with $1\leq i\leq u$ we have $\deg_x(f_i)\leq d$; note: $f$ may contain a period $1$ factor.
Irreducible polynomials $f,g\in\FF[x]$ are called \emph{$\sigma$-equivalent} if there is a $k\in\ZZ$ with $\sigma^k(f)/g\in\FF$. Otherwise, they are called \emph{$\sigma$-coprime}.	

We introduce \rpisiE-extensions~\cite{DR1,DR3} to model, e.g., $(-1)^n$.

\begin{definition}\label{Def:APSExt}
\normalfont
	A difference ring extension $(\EE,\sigma)$ of a difference ring $\dfield{\AR}{\sigma}$ is called an \emph{\rpisiE-extension} if 
	$\AR=\AR_0\leq\AR_1\leq\dots\leq\AR_e=\EE$
	is a tower of ring extensions with $\const{\EE}{\sigma}=\const{\AR}{\sigma}$ where for all $1\leq i\leq
	e$ one of the following holds:
	\begin{itemize}[topsep=0pt, partopsep=0pt, leftmargin=7pt, itemindent=5pt]
		\item $\AR_i=\AR_{i-1}[t_i]$ is a ring extension subject to the relation $t_i^{\nu}=1$ for some $\nu>1$ where $\frac{\sigma(t_i)}{t_i}\in(\AR_{i-1})^*$ is a primitive  $\nu$th root of unity ($t_i$ is called an \emph{\rE-monomial}, and and we define $\nu=\ord(t_i)$);
		\item $\AR_i=\AR_{i-1}[t_i,t_i^{-1}]$ is a Laurent polynomial ring extension with $\frac{\sigma(t_i)}{t_i}\in(\AR_{i-1})^*$ ($t_i$ is called a \emph{\piE-monomial});
		\item $\AR_i=\AR_{i-1}[t_i]$ is a polynomial ring extension with $\sigma(t_i)-t_i\in\AR_{i-1}$ ($t_i$ is called an \emph{\sigmaE-monomial}).
	\end{itemize}
	Depending on the occurrences of the \rpisiE-monomials such an extension is also called a \emph{\rE-/\piE-/\sigmaE-/$R\Pi$-/$R\Sigma$-/\pisiE-extension}.\\
	$\dfield{\EE}{\sigma}$ is called a \emph{simple \rpisiE-ring extension} of $\dfield{\AR}{\sigma}$ if for all \rpiE-monomials $t_i$ we have $\frac{\sigma(t_{i})}{t_{i}}=u\,t_1^{m_1}\dots t_{i-1}^{m_{i-1}}$ with $u\in\AR^*$ and $m_j=0$ if $t_j$ is a \sigmaE-monomial. If $t_i$ is an \rE-monomial, we require in addition that $u$ is a root of unity and $m_j=0$ if $t_j$ is an \pisiE-monomial. 
\end{definition}

\begin{example}\label{Exp:DR}
	Take the difference field $\dfield{\QQ(x)}{\sigma}$ with $\sigma(x)=x+1$. Since $\const{\QQ(x)}{\sigma}=\QQ$, it is a \pisiE-field over $\QQ$. We introduce the following \rpisiE-extensions $\dfield{\EE}{\sigma}$ over $\dfield{\QQ(x)}{\sigma}$, i.e., $\const{\EE}{\sigma}=\QQ$; for algorithmic techniques that verify this property we refer to~\cite{Karr:81,DR1}.
\begin{enumerate}[topsep=0pt, partopsep=0pt, leftmargin=7pt, itemindent=5pt,label={\arabic*.}]
\item $\dfield{\EE}{\sigma}$ with the polynomial ring $\EE=\QQ(x)[h_1][h_2]$, $\sigma(h_1)=h_1+\frac{1}{1+x}$ and
	$\sigma(h_3)=h_3+\frac{1}{(1+x)^3}$ is a simple \sigmaE-extention of $\dfield{\KK(x)}{\sigma}$.
\item Take the ring $\EE_0=\KK(x)[z]$ subject to the relation $z^2=1$ and define on top the Laurent polynomial ring $\EE=\EE_0[\tau_1,\tau_1^{-1}][\tau_2,\tau_2^{-1}]$. Then $\dfield{\EE}{\sigma}$ with $\sigma(z)=-z$,	
	$\sigma(\tau_1)=(x+1)\tau_1$ and $\sigma(\tau_2)=(x+1)\tau_1$
	is a simple \rpisiE-extension of $\dfield{\QQ(x)}{\sigma}$.
\item Take the polynomial ring $\EE=\EE_0[h_1]$. Then $\dfield{\EE}{\sigma}$ with 
	$\sigma(z)=-z$ and $\sigma(h_1)=h_1+\frac{-z}{1+x}$ is a simple \rpisiE-extension of $\dfield{\QQ(x)}{\sigma}$.
\end{enumerate}
\end{example}

For convenience we use $\AR\lr{t}$ with three different meanings: it is the 
ring $\AR[t]$ subject to the relation $t^{\nu}=1$ if $t$ is an \rE-monomial of order $\nu$, it is the polynomial ring $\AR[t]$ if $t$ is a \sigmaE-monomial, or it is the Laurent polynomial ring $\AR[t,t^{-1}]$ if $t$ is a \piE-monomial.\\
Let $\dfield{\EE}{\sigma}$ be a simple \rpisiE-ring extension of $\dfield{\AR}{\sigma}$ with $\EE=\AR\langle t_{1}\rangle\dots\langle t_{e}\rangle$. The elements in $\EE$ are spanned over the power products
$\vect{t}^{\vect{n}}=t_1^{n_1}\dots t_e^{n_e}\in\EE$
with $\vect{n}=(n_1,\dots,n_e)\in\ZZ^e$ where $n_i\geq0$ if $t_i$ is a \sigmaE-monomial. If \emph{$\vect{n}$ is reduced}, i.e., if $0\leq i<\ord(t_i)$ in case that $t_i$ is an \rE-monomial, the power products are uniquely given. In particular, $\vect{t}^{\vect{n}}\in\AR$ implies $\vect{n}=\vect{0}$. Furthermore, one can reorder the generators in $\EE$ such that first \rE-monomials, then \piE-monomials and finally \sigmaE-monomials are adjoined.

Subsequently, let $\AR=\FF(x)$ where $\dfield{\FF(x)}{\sigma}$ is a \pisiE-field extension of $\dfield{\FF}{\sigma}$. 
Let
$f=\sum_{\vect{i}\in\ZZ^e}f_\vect{i}\,\vect{t}^{\vect{i}}\in\EE$  with $f_\vect{i}=\frac{p_\vect{i}}{q_\vect{i}}$ where the polynomials $p_\vect{i},q_\vect{i}\in\FF[x]$ are coprime. Define $q=\lcm_\vect{i}(q_\vect{i})\in\FF[x]$ being monic. 
Then we say that $f=\frac{h}{q}$ with $h=\sum_{\vect{i}\in\ZZ^e}f'_\vect{i}\,\vect{t}^{\vect{i}}\in\FF[x]\lr{t_1}\dots\lr{t_e}$ and $f'_\vect{i}=f_{\vect{i}}\,q/q_i\in\FF[x]$ is in \emph{reduced representation}, and we denote $q$ by $\den(f)$.
Subsequently, we will use the following properties: if $\den(h)$, $\den(g)$ with $h,g\in\EE$ have $x$-degrees{$\leq d$}, then $\den(h\smash{{}^+_\bullet}g)$ have $x$-degrees{$\leq d$}. Further, if $\den(h)$ has $x$-degree{$\leq d$} but not $\den(g)$, then $\den(h\smash{{}^+_\bullet}g)$ do not have $x$-degrees{$\leq d$}.

Finally, we refine simple \rpisiE-exten\-sions further as follows.

\begin{definition}\label{Def:DegreeOptimalExt}
\normalfont
Let $\dfield{\FF(x)}{\sigma}$ be a \pisiE-field extension of $\dfield{\FF}{\sigma}$ and let $\dfield{\EE}{\sigma}$ be a simple \rpisiE-extension of $\dfield{\FF(x)}{\sigma}$ with $\EE=\FF(x)\lr{t_1}\dots\lr{t_e}$. Then this \emph{extension has $x$-degree{$\leq d$}} with $d\in\NN$ 
if for all \pisiE-monomials $t_i$ one of the following properties hold:
\begin{itemize}[topsep=0pt, partopsep=0pt, leftmargin=7pt, itemindent=5pt]
\item If $t_i$ is a \piE-monomial, then $\frac{\sigma(t_i)}{t_i}=u\,t^{n_1}\dots t^{n_{i-1}}$ where $n_j\in\ZZ$ and $u\in\FF(x)^*$ has $x$-degree{$\leq d$}.
\item If $t_i$ is a \sigmaE-monomial, $\sigma(t_i)-t_i=f$ where $\den(f)$ has $x$-degree{$\leq d$}.
\end{itemize}
\end{definition}

\noindent With $\FF=\QQ$ all the difference rings in Example~\ref{Exp:DR} have $x$-degree$\leq 1$.
%We depend on the following properties within such extensions.

\begin{lemma}\label{Lemma:BoundLargeFactors}
Let $\dfield{\FF(x)}{\sigma}$ be a \pisiE-field extension of $\dfield{\FF}{\sigma}$ and let $\dfield{\EE}{\sigma}$ be a simple \rpisiE-extension of $\dfield{\FF(x)}{\sigma}$ with $x$-degree$\leq d\in\NN$. 
\begin{enumerate}[topsep=0pt, partopsep=0pt, leftmargin=7pt, itemindent=5pt,label={\arabic*.}]
\item Let $f\in\EE$ such that $\den(f)$ has $x$-degree{$\leq d$}. Then $\den(\sigma^k(f))$ has $x$-degree{$\leq d$} for any $k\in\ZZ$.
\item Let $f\in\EE$ with $\den(f)=b\,c$  where $c\in\FF[x]$ contains precisely the irreducible period $0$ factors with $x$-degrees larger than $d$ and $b$ has $x$-degree{$\leq d$}. Then for any $k\in\ZZ$ we have  $\den(\sigma^k(f))=\sigma^k(c)B$ for some $B\in\FF[x]$ which has $x$-degree{$\leq d$}.
\item For $g\in\EE$ and a period $0$ irreducible $q\in\FF[x]$ with $\deg_x(q)>d$ the following holds:\\
(i) If $q\mid\den(g)$ and $\sigma(q)\nmid\den(g)$, then $\sigma(q)\mid\den(\sigma(g)-g)$.\\
(ii) If $\sigma^k(q)\nmid\den(g)$ for any $k\in\ZZ$, then $\sigma^k(q)\nmid\den(\sigma(g)-g)$ for any $k\in\ZZ$. 
\end{enumerate}
\end{lemma} 
\begin{proof}
(1) We show statement 1 by induction on $e$. 
The base case $e=0$ obviously holds. Now suppose that the lemma holds for $e-1$ extensions and consider the next \rpisiE-monomial $t_e$ with $\sigma(t_e)=\alpha\,t_e+\beta$. 
If $t_e$ is an \rpiE-monomial, then $\beta=0$ and $\alpha=u\,m$ with $u\in\FF(x)^*$, $m=t_1^{z_1}\dots t_{e-1}^{z_{e-1}}$ with $z_i\in\ZZ$; here $z_i=0$ for all $1\leq i<e$ if $t_i$ is a \sigmaE-monomial. Note that $\sigma^k(t_e)=\alpha_k t_e+\beta_k$ with $\beta_k=0$ and $\alpha_k=\prod_{i=0}^{k-1}\sigma^i(u\,m)$ if $k\geq0$ and $\alpha_k=\prod_{i=1}^{-k}\sigma^{-i}(u^{-1}\,m^{-1})$ if $k<0$. Since $u$ has $x$-degree{$\leq d$} (it $t_e$ is an \rE-monomial, $u\in\FF^*$ is a root of unity), the induction assumption can be applied and it follows that $\den(\alpha_k)$ has $x$-degree{$\leq d$} for any $k\in\ZZ$. Otherwise, suppose that $t_e$ is a \sigmaE-monomial with $\alpha=1$ and $\beta=\sigma(t_e)-t_e\in\FF(x)\lr{t_1}\dots\lr{t_{e-1}}$. Then $\sigma^k(t_e)=\alpha_k\,t_e+\beta_k$ with $\alpha_k=1$ and $\beta_k=\sum_{i=0}^{k-1}\sigma^i(\beta)$ if $k\geq0$ and $\beta_k=-\sum_{i=1}^{k}\sigma^{-i}(\beta)$ if $k<0$. Since $\den(\beta)$ has $x$-degree{$\leq d$}, we can apply again the induction assumption and $\den(\beta_k)$ has $x$-degree$<d$ for any $k\in\ZZ$.
Now consider $f=\sum_{i}f_i\,t_e^i\in\EE$ with $f_i\in\FF(x)\lr{t_1}\dots\lr{t_{e-1}}$. Then $\sigma^k(f)=\sum_{i}\sigma^k(f_i)(\alpha_k\,t_e+\beta_k)^i$
where all components have $x$-degrees{$\leq d$}. 
Thus $\sigma^k(f)$ has $x$-degree{$\leq d$}.\\
(2) Let $f=\frac{a}{b\,c}$ with $a\in\FF[x]\lr{t_1}\dots\lr{t_e}$ and $\den(f)=b\,c$ as claimed in statement~2.
Let $k\in\ZZ$ and consider $\sigma^k(f)=\frac{A}{B\,C}$ with $\den(\sigma^k(f))=B\,C$ where $C$ contains precisely the period $0$ irreducible factors having $x$-degrees larger than $d$ and $B$ has $x$-degree{$\leq d$}.  By statement~1 it follows that $\sigma^{k}(a)=\frac{a'}{b'}$ with $a'\in\FF[x]\lr{t_1}\dots\lr{t_e}$ and $b'=\den(\sigma^{k}(a))\in\FF[x]$ has $x$-degree{$\leq d$}. Thus $\sigma^k(f)=\frac{\sigma^k(a)}{\sigma^k(b)\,\sigma^k(c)}=\frac{a'}{b'\sigma^{k}(b)\sigma^{k}(c)}$ where $\sigma^{k}(c)$ contains all irreducible period $0$ factors of $\den(f)$ whose degree is larger than $d$ and $b'\sigma^{k}(b)$ has $x$-degree{$\leq d$}. Note that cancellation might happen. However, $C\mid\sigma^k(c)$. Now consider $f=\frac{\sigma^{-k}(A)}{\sigma^{-k}(B)\sigma^{-k}(C)}$. Similarly, we get $f=\frac{A'}{B'\sigma^{-k}(C)}$ with $A'\in\FF[x]\lr{t_1}\dots\lr{t_e}$ and $B'\in\FF[x]$ has $x$-degree{$\leq d$}. This implies that $c\mid\sigma^{-k}(C)$ and thus $\sigma(c)\mid C$. Consequently $c=C\,u$ for some $u\in\FF^*$ and the statement is proven.\\
(3) Write $g=\frac{a}{b\,c}$ with $a\in\FF[x]\lr{t_1}\dots\lr{t_e}$ and $\den(g)=b\,c$ with $b,c\in\FF[x]$ were $b$ has $x$-degree{$\leq d$} and $c$ contains all period $0$ irreducible factors whose $x$-degrees are larger than $d$. Then $\sigma(g)=\frac{A}{B\,\sigma(c)}$ with $\den(\sigma(g))=B\,\sigma(c)$ where $B\in\FF[x]$ has $x$-degree{$\leq d$} by statement~2. (i) Suppose $q\mid\den(g)$. Thus $q\mid c$, hence $\sigma(q)\mid\sigma(c)$ and therefore $\sigma(q)\mid\den(\sigma(g))$. By the second assumption $\sigma(q)\nmid\den(g)$ it follows that $\sigma(q)\mid\den(\sigma(g)-g)$.\\ 
(ii) If $\sigma^k(q)$ is no factor of $\den(g)$ for any $k\in\ZZ$, then it is no factor of $c$ and thus of $\sigma(c)$. Consequently it is not a factor in $\den(\sigma(g))$. In particular, it cannot be a factor in $\den(\sigma(g)-g)$.
\end{proof}

\section{Refined representations}\label{Sec:RefinedRep}

We start with the following definition and lemmas to get a normalized representation of the denominator of a given input summand.

\begin{definition}\label{Def:SpecialSetsforD}
	\normalfont
	Let $\dfield{\FF(x)}{\sigma}$ be a \pisiE-field extension of $\dfield{\FF}{\sigma}$ and $d\in\NN$. 
	We call a finite set $Q\subseteq\FF[x]$ of monic irreducible polynomials a \emph{$(d,x)$-set} if the degrees are larger than $d$, they have period $0$ and are pairwise $\sigma$-coprime.
	Let $f\in\FF[x]\setminus\{0\}$. Then a $(d,x)$-set $Q$ is called \emph{$(d,f)$-complete} if for any irreducible factor $h$ of $f$ with $\deg_x(h)>d$ there are $q\in Q$ and $k\in\ZZ$ with $\sigma^k(q)/h\in\FF$.
\end{definition}

In the following we require that one can solve Problem~SE; for algorithmic details see Thm.~\ref{Thm:Computable} below.
\medskip

\noindent\textbf{Problem SE in $\dfield{\FF(x)}{\sigma}$} (Shift Equivalence) \textit{Given} a \pisiE-field extension $\dfield{\FF(x)}{\sigma}$ of $\dfield{\FF}{\sigma}$ and irreducible $f,g\in\FF[x]$. \textit{Decide} constructively if there is a $k\in\ZZ$ with $\sigma^k(f)/g\in\FF$.

\begin{lemma}\label{Lemma:UpdateCompleteSet}
	Let $\dfield{\FF(x)}{\sigma}$ be a \pisiE-field extension of $\dfield{\FF}{\sigma}$ in which one can solve Problem~SE and can factorize polynomials. Let $f\in\FF[x]\setminus\{0\}$, $d\in\NN$ and $Q\subseteq\FF[x]$ be an $(d,x)$-set. Then one can compute a set $Q'\supseteq Q$ which is $(d,f)$-complete.
\end{lemma}
\begin{proof}
	Compute all irreducible, pairwise coprime, period $0$ factors $f_1,\dots,f_m\in\FF[x]$ of $f$ with $\deg_x(f_i)>d$. If $m=0$, $Q$ is the desired result. Otherwise, set $Q':=Q$ and proceed for each $i=1\dots m$ and check if there is a $q\in Q'$ and $k\in\ZZ$ with $\sigma^k(f_i)/q$; if there is none, set $Q':=Q'\cup\{f_i\}$. The obtained $Q'$ is $(d,f)$-complete.
\end{proof}

Given these notions, we obtain the following representation; it can be considered as a variant of partial fraction decomposition and is connected to constructions given~\cite{Abramov:75,Paule:95,Abramov:03,BChen:05,Schneider:07d,CSFFL:15}.

\begin{lemma}\label{Lemma:TransformToReducedForm}
Let $\dfield{\FF(x)}{\sigma}$ be a \pisiE-field extension of $\dfield{\FF}{\sigma}$ and let $\dfield{\EE}{\sigma}$ be a simple \rpisiE-extension of $\dfield{\FF(x)}{\sigma}$ with $x$-degree{$\leq d$} with $d\in\NN$ and $\EE=\FF(x)\lr{t_1}\dots\lr{t_e}$. Let $f\in\EE$ and let $Q=\{q_1,\dots,q_r\}$ be $(d,\den(f))$-complete. Then there are $f',g\in\EE$ s.t.\
\begin{equation}\label{Equ:RefinedTele}
\sigma(g)-g+f'=f
\end{equation}
where $f'$ can be written in the \emph{$\sigma$-reduced form}
\begin{equation}\label{Equ:NormalForm}
f'=\frac{p_1}{q_1^{n_1}}+\frac{p_r}{q_r^{n_r}}+\frac{p}{q}
\end{equation}
with the following ingredients: 
\begin{enumerate}[topsep=0pt, partopsep=0pt, leftmargin=7pt, itemindent=5pt,label={\arabic*.}]
\item $n_1,\dots,n_r\in\NP$,
\item $q\in\FF[x]\setminus\{0\}$ with $x$-degree {$\leq d$}, 
\item $p\in\FF[x]\lr{t_1}\dots\lr{t_e}$,
%\item $q_1,\dots,q_r\in\FF[x]\setminus\{0\}$ are monic irreducible polynomials pairwise $\sigma$-prime with $\deg_x(q_i)>d$, 
\item and $p_1,\dots,p_r\in\FF[x]\lr{t_1}\dots\lr{t_e}$ with $\deg_x(p_i)<\deg_x(q_i)\,n_i$. 
\end{enumerate}
If one can factorize polynomials in $\FF[x]$ and can solve Problem~SE in a computable $\dfield{\FF(x)}{\sigma}$, then $g$ and $f'$ with~\eqref{Equ:NormalForm} can be computed.
\end{lemma}

\begin{proof}
Write $f=\frac{a}{b}$ with $a\in\FF[x]\lr{t_1}\dots\lr{t_e}$ and $b\in\FF[x]\setminus\{0\}$ monic in reduced representation. If $\FF$ is computable, this can be accomplished with the Euclidean algorithm. In particular, write $b=q'\,b'$ where $q'\in\FF[x]\setminus\{0\}$ has $x$-degree{$\leq d$} and where
$b'=v_1^{n_1}\dots v_r^{n_r}\in\FF[x]$ with $v_i\in\FF[x]\setminus\{0\}$ are the monic irreducible and period $0$ factors with $\deg_x(v_i)>d$. Note that $Q$ is $(d,b')$-complete.
If $b'=1$, we can take $p=a$, $q=b$, $r=0$ and $g=0$, and we are done. Otherwise, take $s,t\in\FF[x]$ such that $1=s\,b'+t\,q'$; since $\gcd(b',q')=1$, such $s$ and $t$ exist and can be calculated by the extended Euclidean algorithm if $\FF$ is computable. Hence 
$\frac{a}{b}=\frac{s\,a}{q'}+\frac{t\,a}{b'}$. Now we repeat this tactic to $\frac{s\,a}{b'}$ iteratively to separate the coprime factors $v_i^{n_i}$ in the denominator of $b'$ and get 

\vspace*{-0.2cm}

\begin{equation}\label{Equ:PFDStep}
\frac{a}{b}=\frac{s\,a}{q'}+\frac{u_1}{v_1^{n_1}}+\dots+\frac{u_r}{v_{r}^{n_r}}
\end{equation}

\vspace*{-0.1cm}

\noindent with $u_i\in\FF[x]\lr{t_1}\dots\lr{t_e}$. W.l.o.g.\ suppose that $v_1,\dots,v_k$ are all those factors that are $\sigma$-equivalent to $q_1\in Q$. Hence for all $1\leq i\leq k$, $c_i:=\frac{\sigma^{s_i}(q_1)}{v_i}\in\FF$ for some uniquely determined $s_i\in\ZZ$; see Lemma~\ref{Lemma:ShiftProp}.2. Define $f'_i,\gamma_i\in\EE$ with $f'_i=\frac{\sigma^{-s_i}(u_i\,c_i^{n_i})}{q_1^{n_i}}$ and
%$$\gamma_i:=\begin{cases}
%\sum_{j=0}^{s_i-1}\sigma^j(\frac{\sigma^{-s_i}(u_i\,c_i^{n_i})}{q_1^{n_i}})\in\EE&\text{if $s_i\geq0$}\\
%-\sum_{j=1}^{-s_i}\sigma^{-j}(\frac{\sigma^{-s_i}(u_i\,c_i^{n_i})}{q_1^{n_i}})\in\EE&\text{if $s_i<0$.}
%\end{cases}$$
$\gamma_i=\sum_{j=0}^{s_i-1}\sigma^j(\frac{\sigma^{-s_i}(u_i\,c_i^{n_i})}{q_1^{n_i}})$ if $s_i\geq0$ and
$\gamma_i=-\sum_{j=1}^{-s_i}\sigma^{-j}(\frac{\sigma^{-s_i}(u_i\,c_i^{n_i})}{q_1^{n_i}})$ if $s_i<0$.
Then by telescoping and $\sigma^{s_i}(q_1)=c_i\,v_i$ we get 
$$\sigma(\gamma_i)-\gamma_i+f'_i=\sigma^{s_i}\big(\sigma^{-s_i}(u_i\,c_i^{n_i})/q_1^{n_i}\big)=\frac{u_i\,c_i^{n_i}}{\sigma^{s_i}(q_1)^{n_i}}=\frac{u_i}{v_i^{n_i}}.$$
Since $u_i\,c_i^{n_i}\in\FF[x]\lr{t_1}\dots\lr{t_e}$, it follows that $\den(\sigma^{-s_i}(u_i\,c_i^{n_i}))$ has no irreducible factors with $x$-degrees larger than $d$ by Lemma~\ref{Lemma:BoundLargeFactors}.1. Thus we can write
$f'_i=\frac{\alpha_i}{\beta_i\,q_1^{n_1}}$ with $\alpha_i\in\FF[x]\lr{t_1}\dots\lr{t_e}$ and $\beta_i\in\FF[x]\setminus\{0\}$ whose irreducible factors have $x$-degrees{$\leq d$}. Since $\gcd(q_1^{n_1},\beta_i)=1$, we can write  $s'_i\,\beta_i+t'_i\,q_1^{n_1}=1$ with $s'_i,t'_i\in\FF[x]$ (and can again compute it if $\FF$ is computable). Hence $f'_i=\phi_i+h_i$ with $\phi_i=\frac{\alpha_i\,s'_i}{q_1^{n_1}}$ and $h_i=\frac{\alpha_i\,t'_i}{\beta_i}$.
Now take $g':=\gamma_1+\dots+\gamma_k\in\EE$. Further let $\frac{H_1}{H_2}=h_1+\dots+h_k$ where $H_1\in\FF[x]\langle t_1,\dots,t_e\rangle$ and $H_2\in\FF[x]\setminus\{0\}$ which has $x$-degree{$\leq d$}. In addition, define $w_1\in\FF[x]\langle t_1,\dots,t_e\rangle$ with $\frac{w_1}{q_1^{n_1}}=\phi_1+\dots+\phi_k$. Then 
$\sigma(g')-g'+\frac{w_1}{q_1^{n_1}}+\frac{H_1}{H_2}=\frac{u_1}{v_1^{n_1}}+\dots+\frac{u_k}{v_k^{n_k}}.$
Since $q_1^{n_1}\in\FF[x]\setminus\{0\}$, the leading coefficient of $q_1^{n_1}$ is a unit. Thus we can compute $w_1=\tilde{q}_1\,q_1^{n_1}+p_1$ with $p_1,\tilde{q}_1\in\FF[x]\langle t_1,\dots,t_e\rangle$ with $\deg_x(p_1)<\deg_x(q_1^n)$ by polynomial division (and considering $x$ as the top variable). This gives
$\frac{w_1}{q_1^{n_1}}=\frac{p_1}{q_1^{n_1}}+\tilde{q}_1$. Define $p''\in\FF[x]\langle t_1,\dots,t_e\rangle$ and $q''\in\FF[x]\setminus\{0\}$ with
$\frac{p''}{q''}=\frac{a\,a}{q'}+\tilde{q}_1+\frac{H_1}{H_2}$; note: $q''$ has only factors with $x$-degree{$\leq d$}. Thus plugging the ingredients into~\eqref{Equ:PFDStep} gives

\vspace*{-0.6cm}

$$\frac{a}{b}=\sigma(g')-g'+\frac{p''}{q''}+\frac{p_1}{q_1^n}+\overbrace{\frac{u_{k+1}}{v_{k+1}^{m_{k+1}}}+\dots+\frac{u_r}{v_{r}^{m_r}}}^{R}.$$
Repeating this transformation to $R$ produces the desired representation~\eqref{Equ:NormalForm}. If one can factorize polynomials in $\FF[x]$ and one can solve Problem~SE in $\dfield{\FF(x)}{\sigma}$, this representation can be calculated.
\end{proof}

\begin{example}\label{Exp:SigmaReduceForm}
1. Take the difference ring $\dfield{\EE}{\sigma}$ from Example~\ref{Exp:DR}.1. Here we can rephrase the summand on the left-hand side of~\eqref{Equ:HarmonicSumId} with
\begin{equation}\label{Equ:fFromHarmonicSum}
f=\frac{-2+x}{10 (
	1+x^2)}
+\frac{h_1 (
	1-4 x-2 x^2)}{10 (
	1+x^2
	)
	(2+2 x+x^2)}
+\frac{h_3 (
	1-4 x-2 x^2)}{5 (
	1+x^2
	)
	(2+2 x+x^2)}\in\EE.
\end{equation}
$Q=\{q_1\}$ with $q_1=1+x^2$ is $(1,f)$-complete. We can compute
\begin{equation}\label{Equ:gFromHarmonicSum}
g=\frac{h_3 (1+2 x)}{10 (
	1+x^2)}
+\frac{h_3 (1+2 x)}{5 (
	1+x^2)}
-\frac{(1+2 x) (
	2+x^2)}{10 x^3 (
	1+x^2)}
\end{equation}
and $f'=\frac{p_1}{q_1}+\frac{p}{q}$ with $p_1=0$ and $\frac{p}{q}=\frac{-2-4 x+x^2}{10 x^3}$ such that~\eqref{Equ:RefinedTele} holds. \\
2. Take the difference ring $\dfield{\EE}{\sigma}$ from Example~\ref{Exp:DR}.3 and consider
\begin{equation}\label{Exp:VectorF}
\vect{f}=(f_1,f_2,f_3)=\Big(\frac{h_1(
	x
	+z
	+x^2 z)}{x(1+x^2)},\frac{h_1}{2+2 x+x^2},\frac{x z}{1+x^2}\Big)\in\EE^3.
\end{equation}
$Q=\{q_1\}$ with $q_1=1+x^2$ is $(1,f_i)$-complete for $1\leq i\leq 3$. For $\sigma(g_i)-g_i+f'_i=f_i$ we get $g_i\in\EE$ and the $\sigma$-reduced form $f'_i=\frac{p_{i,1}}{q_1}+\rho_i$ with $\rho_i=p'_i/q'_i$.
Namely,	
$p_{1,1}=h_1$, $\rho_1=\frac{h_1 z}{x}$, $g_1=0$, and
$p_{2,1}=h_1+x z$, $\rho_2=-\frac{z}{x}$, $g_2=\frac{h_1 x-z}{x(1+x^2)}$, and
$p_{3,1}=x z$, $\rho_3= 0$, $g_3=0$. 
\end{example}

%We need the following properties in simple \rpiE-extensions.

\begin{lemma}\label{Lemma:PiCoeff}
Let $\dfield{\EE}{\sigma}$ be a simple \rpiE-extension of $\dfield{\HH}{\sigma}$ with $\EE=\HH\lr{t_1}\dots\lr{t_e}$. Let $f=\sum_{\vect{i}\in\ZZ^e}f_\vect{i}\,\vect{t}^{\vect{i}}\in\EE$ and $g=\sum_{\vect{i}\in\ZZ^e}g_\vect{i}\,\vect{t}^{\vect{i}}\in\EE$ with $\sigma(g)-g=f$. Then for each reduced $\vect{i}\in\ZZ^e$ there is a unique reduced $\vect{j}\in\ZZ^e$ with $u=\frac{\sigma(\vect{t}^{\vect{i}})}{\vect{t}^{\vect{j}}}\in\HH^*$. Conversely, for each reduced $\vect{j}\in\ZZ^e$ there is a unique reduced $\vect{i}\in\ZZ^e$ 
with $u=\frac{\sigma(\vect{t}^{\vect{i}})}{\vect{t}^{\vect{j}}}\in\HH^*$.
For such a tuple $(\vect{i},\vect{j})$ with $u=\frac{\sigma(\vect{t}^{\vect{i}})}{\vect{t}^{\vect{j}}}\in\HH^*$ we have $u\,\sigma(g_{\vect{i}})-g_{\vect{j}}=f_{\vect{j}}$. 
\end{lemma}
\begin{proof}
Let $\vect{i}\in\ZZ^e$ be reduced and take $h=\sigma(\vect{t}^{\vect{i}})$. By definition we have $h=u\,\vect{t}^{\vect{j}}$ with $u\in\HH^*$ and a reduced $\vect{j}\in\ZZ^e$, i.e., $\frac{\sigma(\vect{t}^{\vect{i}})}{\vect{t}^{\vect{j}}}=u\in\HH^*$.  
Suppose that $\frac{\sigma(\vect{t}^{\vect{i}})}{\vect{t}^{\vect{k}}}=u'\in\HH^*$ with another reduced $\vect{k}\in\ZZ^e$. Then $\vect{t}^{\vect{j}-\vect{k}}=u'/u\in\HH^*$ which implies that $\vect{j}=\vect{k}$, i.e., $\vect{j}$ is uniquely determined.
Similarly, let $\vect{j}\in\ZZ^e$ be reduced and take $h=\sigma^{-1}(\vect{t}^{\vect{j}})$. By definition we have $h=u'\,\vect{t}^{\vect{i}}$ with $u'\in\HH^*$ and $\vect{i}\in\ZZ^e$ reduced, i.e.,  $\vect{t}^{\vect{i}}/\sigma^{-1}(\vect{t}^{\vect{j}})=1/u'\in\HH^*$ and thus 
$\frac{\sigma(\vect{t}^{\vect{i}})}{\vect{t}^{\vect{j}}}=u$ with $u=\sigma(1/u')\in\HH^*$. Further,
suppose that $\frac{\sigma(\vect{t}^{\vect{k}})}{\vect{t}^{\vect{j}}}=u'\in\HH^*$ with another reduced $\vect{k}\in\ZZ^e$. Then $\sigma(\vect{t}^{\vect{i}-\vect{k}})=u/u'$ and thus $\vect{t}^{\vect{i}-\vect{k}}=\sigma^{-1}(u/u')\in\HH^*$. This implies $\vect{i}=\vect{k}$ and proves the uniqueness of $\vect{i}$.
Now take such a tuple $(\vect{i},\vect{j})$ of reduced elements with $u=\frac{\sigma(\vect{t}^{\vect{i}})}{\vect{t}^{\vect{j}}}\in\HH^*$.
By coefficient comparison in $\sigma(g)-g=f$ w.r.t.\ $\vect{t}^{\vect{j}}$ we get
  $\sigma(g_{\vect{i}}\,\vect{t}^{\vect{i}})-g_{\vect{j}}\,\vect{t}^{\vect{j}}=f_{\vect{j}}\,\vect{t}^{\vect{j}}$. With $\sigma(g_{\vect{i}}\,\vect{t}^{\vect{i}})=\sigma(g_{\vect{i}})\,u\,\vect{t}^{\vect{j}}$ and dividing through $\vect{t}^{\vect{j}}$ we get $u\,\sigma(g_{\vect{i}})-g_{\vect{j}}=f_{\vect{j}}$. 
\end{proof}

\begin{lemma}\label{Lemma:BaseCasePiExt}
Let $\dfield{\FF(x)}{\sigma}$ be a \pisiE-field extension of $\dfield{\FF}{\sigma}$ and $\dfield{\EE}{\sigma}$ be a simple \rpiE-extension of $\dfield{\FF(x)}{\sigma}$ with $x$-degree{$\leq d$}. Let $f\in\EE$ with $\den(f)=v\,q^n$ where $n\in\NP$,
$q\in\FF[x]$ is an irreducible period $0$ factor with $\deg_x(v)>d$ and $v\in\FF[x]$ does not contain any factor which is $\sigma$-equivalent to $q$.
Then there is no $g\in\EE$ with
$\sigma(g)-g=f.$
\end{lemma}
\begin{proof}
Suppose that there is such a $g=\sum_{\vect{i}\in\ZZ^e}g_{\vect{i}}\,\vect{t}^{\vect{i}}\in\EE$  
with $g_{\vect{i}}=\frac{\gamma_{\vect{i}}}{\delta_{\vect{i}}}$ where $\gamma_{\vect{i}}\in\FF[x],\delta_{\vect{i}}\in\FF[x]\setminus\{0\}$ with $\gcd(\gamma_{\vect{i}},\delta_{\vect{i}})=1$.
Write $f=\sum_{\vect{i}\in\ZZ^e}f_\vect{i}\,\vect{t}^{\vect{i}}\in\EE$ with $f_\vect{i}=\frac{a_\vect{i}}{b_\vect{i}}$ with $a_\vect{i},b_\vect{i}\in\FF[x]$ and $\gcd(a_\vect{i},b_\vect{i})=1$ where $b_\vect{i}$ may contain $q$ as factor but not $\sigma^k(q)$ with $k\neq0$.
There must be a denominator $\delta_{\vect{i}}$ that contains $\sigma^{\lambda}(q)$ for some $\lambda\in\ZZ$. Otherwise, we conclude with Lemma~\ref{Lemma:BoundLargeFactors}.3.(ii) that $q\nmid\den(\sigma(g)-g)$, a contradiction. Among all $g_{\vect{j}}\neq0$ take $\vect{j}\in\ZZ^e$ such that $\sigma^{\lambda}(q)\mid\delta_{\vect{j}}$ with $\lambda\in\ZZ$ maximal.  
By Lemma~\ref{Lemma:PiCoeff} we can take $\vect{i}\in\ZZ^e$ with $u=\frac{\sigma(\vect{t}^{\vect{j}})}{\vect{t}^{\vect{i}}}\in\HH^*$ and  $u\,\sigma(g_{\vect{j}})-g_{\vect{i}}=f_{\vect{i}}$. Note that $u=v\,t_1^{z_1}\dots t_e^{z_e}$ with $z_i\in\ZZ$ where $v\in\FF[x]$ has $x$-degree{$\leq d$}. Then $\sigma^{\lambda+1}(q)\mid\sigma(\delta_{\vect{j}})$ but $\sigma^{\lambda+1}(q)\nmid \delta_{\vect{i}}$. Thus 
$\sigma^{\lambda+1}(q)\mid\den(f_{\vect{i}})=b_{\vect{i}}$. Hence $\lambda=-1$ for the maximal choice $\lambda$. Among all $g_{\vect{j}}\neq0$ take $\vect{j}\in\ZZ^e$ such that $\sigma^{\lambda}(q)\mid\delta_{\vect{j}}$ with $\lambda$ minimal. Note that $\lambda<0$ (since the maximal choice is $-1$). By Lemma~\ref{Lemma:PiCoeff} we can take $\vect{j}\in\ZZ^e$ with $u=\frac{\sigma(\vect{t}^{\vect{i}})}{\vect{t}^{\vect{j}}}\in\HH^*$ and  $u\,\sigma(g_{\vect{i}})-g_{\vect{j}}=f_{\vect{j}}$. Similarly, one gets $\sigma^{\lambda}(q)\nmid\den(u\,\sigma(g_{\vect{i}}))$ and we conclude that $\sigma^{\lambda}(q)\mid\den(f_{\vect{j}})=b_{\vect{j}}$ with $\lambda<0$, a contradiction. Thus $g\in\EE$ with $\sigma(g)-g=f$ cannot exist. 
\end{proof}

\noindent Now we can present the main property for \rpisiE-extensions which can be considered as a generalization appearing in~\cite{Abramov:75,Paule:95,Abramov:03,BChen:05,Schneider:07d,CSFFL:15}; there the denominators are split by $\sigma$-equivalent factors.

\begin{proposition}\label{Prop:RemoveBadFactors}
Let $\dfield{\FF(x)}{\sigma}$ be a \pisiE-field extension of $\dfield{\FF}{\sigma}$ and $\dfield{\EE}{\sigma}$ be a simple \rpisiE-extension of $\dfield{\FF(x)}{\sigma}$ with $\EE=\FF(x)\lr{t_1}\dots\lr{t_e}$ and $x$-degree{$\leq d$}. Let $p\in\EE$ and $q\in\FF[x]\setminus\{0\}$ with $x$-degree{$\leq d$}. Let 
$\{q_1,\dots,q_r\}\subseteq\FF[x]$ be a $(d,x)$-set, $n_1,\dots,n_r\in\NP$ and $p_1,\dots,p_r\in\FF[x]\lr{t_1}\dots\lr{t_e}$ with $\deg_x(p_i)\leq\deg_x(q_i)\,n_i$.
If there is a $g\in\EE$ with

\vspace*{-0.4cm}

\begin{equation}\label{Equ:FundamentalThm}
\hspace*{1.5cm}\sigma(g)-g=\frac{p_1}{q_1^{n_1}}+\frac{p_r}{q_r^{n_r}}+\frac{p}{q}
\end{equation}
then $p_1=\dots=p_r=0$ and $\den(g)$ has $x$-degree{$\leq d$}. 
\end{proposition}
\begin{proof}
Write $\EE=\FF(x)\lr{t_1}\dots\lr{t_e}$. W.l.o.g.\ we may assume that all \rE-monomials are adjoined first, \piE-monomials come next and \sigmaE-monomials are adjoint at the end; otherwise we reorder them accordingly.
We prove the proposition by induction of the number of \sigmaE-extensions in $\EE$. In the base case, we assume that $\EE$ is built only by simple \rpiE-extensions and there is a $g\in\EE$ with~\eqref{Equ:FundamentalThm}. In addition assume that there is $i$ such that $p_i$ is nonzero.
Write $\frac{p_i}{q_1^{n_i}}=\frac{p'_i}{q_1^{\mu}}$ in reduced representation with $\mu\leq n_i$. 
Since $\deg_x(p_i)<\deg_x(q_i^{n_i})$, it follows $\mu\geq1$. In particular we can write the right-hand side of~\eqref{Equ:FundamentalThm} in reduced representation with
$f=\frac{p'}{v\,q_i^{\mu}}$ where $p'\in\FF[x]\lr{t_1}\dots\lr{t_e}$ and $v\in\FF[x]$ whose irreducible factors are $\sigma$-coprime with $q_i$. By Lemma~\ref{Lemma:BaseCasePiExt} a solution $g\in\EE$ with~\eqref{Equ:FundamentalThm} is not possible, a contradiction. Hence $p_i=0$ for all $i$, and we get
$\sigma(g)-g=\frac{p}{q}$
with $g\in\EE$. Suppose there is an irreducible period $0$ factor $v$ in $\den(g)$ with $\deg_x(v)>d$. 
By Lemma~\ref{Lemma:ShiftProp}.1 we can take among the $\sigma$-equivalent factors of $v$ in $\den(g)$ that one
which maps to the other factors only by negative $\sigma$-shifts.
By Lemma~\ref{Lemma:BoundLargeFactors}.3.(i), $\sigma(v)\mid\den(\sigma(g)-g)$, a contradiction. Thus $\den(g)$ has $x$-degree{$\leq d$}.

Now consider the simple \rpisiE-extension $\dfield{\FF(x)\lr{t_1}\dots\lr{t_e}}{\sigma}$ of $\dfield{\FF(x)}{\sigma}$ with $x$-degree {$\leq d$} with $\sigma(t_e)=t_e+\beta$ and suppose that the proposition holds for $\dfield{\HH}{\sigma}$ with $\HH=\FF(x)\lr{t_1}\dots\lr{t_{e-1}}$. Let $g\in\FF(x)\lr{t_1}\dots\lr{t_e}$ such that~\eqref{Equ:FundamentalThm} holds. By~\cite[Lemma~7.2]{DR1} it follows that for $b=\deg_{t_e}(f)+1$ we have that $\deg_{t_e}(g)\leq b$. We show the proposition by a second induction on $b$. If $b=0$, it follows that $f$ is free of $t_e$ and the main induction assumption implies the correctness. Now suppose that the proposition holds for a solution where the degree is smaller than $b$. Define $\phi:=\coeff(f,t_e,b)\in\HH$ and $\gamma:=\coeff(g,t_e,b)\in\HH$ being the coefficients of $t_e^b$ in $f$ and $g$. Then $g=\gamma\,t_e^b+w$ with $w\in\FF(x)\langle t_1,\dots,t_{e-1}\rangle[t_e]$ where $\deg_{t_e}(w)<b$. By coefficient comparison it follows that $\sigma(\gamma)-\gamma=\phi$. Define $h_i:=\coeff(p_i,t_e,b)\in\HH$ and $u:=\coeff(p,t_e,b)\in\HH$. Then
$\phi=\frac{h_1}{q_1^{n_1}}+\dots+\frac{h_r}{q_r^{n_r}}+\frac{u}{q}$
holds in $\HH$. Hence by the induction assumption on $\HH$ we conclude that $h_1=\dots=h_r=0$ and $\den(\gamma)$ has $x$-degree{$\leq d$}. Now define $U=\frac{p}{q}-(\sigma(\gamma\,t^b)-\gamma\,t^b)\in\HH[t_e]$. Then by construction 
$\sigma(w)-w=\frac{p_1}{q_1^{n_1}}+\dots+\frac{p_r}{q_r^{n_r}}+U$
and $\deg_{t_e}(U)<b$. Moreover, $\den(\sigma(\gamma\,t^b))$ has $x$-degrees{$\leq d$} by Lemma~\ref{Lemma:BoundLargeFactors}.1.
Since also $q$ and $\den(\gamma)$ have $x$-degrees{$\leq d$}, we conclude that $\den(U)$ has $x$-degree{$\leq d$}. With the second induction hypothesis (induction on $b$) it follows that $p_1=\dots=p_r=0$ and $\den(w)$ has $x$-degree{$\leq d$}. Hence $\den(\gamma\,t^b+w)$ has $x$-degree{$\leq d$}. This completes the proof.
\end{proof}

\noindent In the following we rely on Theorem~\ref{Thm:SigmaExtPara} shown in~\cite[Theorem~7.10]{DR3}; for the field version with $m=1$ see~\cite{Karr:81} and for the general case $m\in\NN$ see~\cite{Schneider:10c}; this result is also related to~\cite{Singer:08}.

\begin{theorem}[\cite{DR3}]\label{Thm:SigmaExtPara}
	Let $\dfield{\AR}{\sigma}$ be a difference ring with constant field $\KK=\const{\AR}{\sigma}$ and let $f_1,\dots,f_m\in\AR$. Then there is a \sigmaE-extension $\dfield{\AR[s_1]\dots[s_m]}{\sigma}$ of $\dfield{\AR}{\sigma}$ with $\sigma(s_i)=s_i+f_i$ for $1\leq i\leq m$ iff there are no $(c_1,\dots,c_m)\in\KK^m\setminus\{\vect{0}\}$ and $h\in\AR$ with~\eqref{Equ:ParaT}. 
\end{theorem}

\noindent Using this result we obtain the following characterization of certain classes of simple \rpisiE-extension. They will be introduced in Def.~\ref{Def:reducedExt} below and will be the basis of our telescoping algorithms.

\begin{theorem}\label{Thm:SigmaCharDeg}
	Let $\dfield{\FF(x)}{\sigma}$ be a \pisiE-field extension of $\dfield{\FF}{\sigma}$ with $\KK=\const{\FF}{\sigma}$ and let $\dfield{\EE}{\sigma}$ be a simple \rpisiE-extension of $\dfield{\FF(x)}{\sigma}$ with $x$-degree{$\leq d$} and $\EE=\FF(x)\lr{t_1}\dots\lr{t_e}$. Let $\{q_1,\dots,q_r\}\subseteq\FF[x]$ be a $(d,x)$-set, $n_1,\dots,n_r\in\NP$, and for $1\leq i\leq r$ and $1\leq j\leq e_i$ with $e_i\geq1$ let $p_{i,j}\in\FF[x]\lr{t_1}\dots\lr{t_e}\setminus\{0\}$ with $\deg_x(p_{i,j})<\deg_x(q_i)n_{i}$. Then the following statements are equivalent.
	\begin{enumerate}[topsep=0pt, partopsep=0pt, leftmargin=7pt, itemindent=5pt,label={\arabic*.}]
		\item $\Ann_{\KK}(p_{i,1},\dots,p_{i,e_i})=\{\vect{0}\}$ for all $1\leq i\leq r$; 
		\item there are no $g\in\EE$ and $c_{i,j}\in\KK$ (not all zero) such that
		\begin{equation}\label{Equ:CreaSolSigma}
			\sigma(g)-g=\sum_{i,j}c_{i,j}\frac{p_{i,j}}{q_i^{n_i}};
		\end{equation}
		\item the difference ring extension
		$\dfield{\SA}{\sigma}$ of $\dfield{\EE}{\sigma}$ with the polynomial ring 
		$$\SA=\EE[s_{1,1},\dots,s_{1,e_1},\dots,s_{r,1},\dots,s_{r,e_r}]$$ and $\sigma(s_{i,j})=s_{i,j}+\frac{p_{i,j}}{q_i^{n_i}}$ is a \sigmaE-extension, i.e., $\const{\SA}{\sigma}=\const{\FF}{\sigma}$.
	\end{enumerate}
\end{theorem}
\begin{proof}
	$(1)\Rightarrow(2)$ Let $c_{i,j}\in\KK$, not all zero, and $g\in\EE$ such that
	$\sigma(g)-g=\sum_{i,j}c_{i,j}\frac{p_{i,j}}{q_i^{n_i}}=\sum_{1\leq i\leq r}\frac{p_i}{q_i^{n_i}}$	with 
	$p_i=\sum_{j=1}^{e_i}c_{i,j}p_{i,j}\in\FF[x]\lr{t_1}\dots\lr{t_e}$
	for $1\leq i\leq r$. Since $\deg_x(p_{i,j})<\deg_x(q_i)n_{i}$, we have $\deg_x(p_i)<\deg_x(q_i)n_{i}$. Hence we can apply Proposition~\ref{Prop:RemoveBadFactors} and it follows that $p_i=0$ for all $i$. 
	By assumption  we can take $i,j$ with $c_{i,j}\neq0$. Thus $\vect{0}\neq (c_{i,1},\dots,c_{i,e_i})\in\Ann_{\KK}(p_{i,1},\dots,p_{i,e_i})$. \\
	$(2)\Rightarrow(1)$ Suppose that there is $i$ with $V_i=\Ann_{\KK}(p_{i,1},\dots,p_{i,e_i})\neq\{\vect{0}\}$. Then we can take $g=0\in\EE$ and $(c_{i,1},\dots,c_{i,e_i})\in V_i\setminus\{\vect{0}\}$ where all other $c_{k,j}0$ are set to zero. This gives~\eqref{Equ:CreaSolSigma}.\\
	$(2)\Leftrightarrow(3)$ follows by Theorem~\ref{Thm:SigmaExtPara}.
\end{proof}

\section{Refined telescoping algorithms}\label{Sec:RefinedTele}

We will assume that certain algorithmic properties are satisfied in the ground field $\dfield{\FF(x)}{\sigma}$. Here we can exploit the following result.

\begin{theorem}[\cite{Karr:81,DR3}]\label{Thm:Computable}
	Let $\KK=A(y_1,\dots,y_{\lambda})$ be a rational function field over an algebraic number field $A$, 
	$\dfield{\FF(x)}{\sigma}$ be a \pisiE-field over $\KK$ and $\dfield{\EE}{\sigma}$ be a simple \rpisiE-extension of $\dfield{\FF}{\sigma}$. Then: 
	\begin{enumerate}[topsep=0pt, partopsep=0pt, leftmargin=7pt, itemindent=5pt,label={\arabic*.}]
		\item One can solve Problem~SE in $\dfield{\FF(x)}{\sigma}$. 
		\item $\dfield{\EE}{\sigma}$ is LA-computable and Problems~T and~PT are solvable in $\EE$.
	\end{enumerate}
\end{theorem}

	Statement~1 of Theorem~\ref{Thm:Computable} follows by~\cite{Karr:81,Ge:93,Schneider:05c} and statement~2 by~\cite{Karr:81,DR3}. 
	More general difference fields $\dfield{\FF}{\sigma}$ can be constructed provided that certain algorithmic properties hold in $\FF$; compare~\cite[Sec.~2.3.3]{DR1}. E.g., one can take \pisiE-field extensions and radical field extensions~\cite{Schneider:07f} over free difference fields~\cite{Schneider:06d,Schneider:06e}.

\noindent In the following we will obtain an enhanced telescoping algorithm that works for the following subclass of simple \rpisiE-extensions; note that these extensions are precisely those which are characterized in Thm.~\ref{Thm:SigmaCharDeg}.

\begin{definition}\label{Def:reducedExt}	
	\normalfont
	Let $\dfield{\FF(x)}{\sigma}$ be a \pisiE-field extension of $\dfield{\FF}{\sigma}$, $d\in\NN$ and 
	$Q\subseteq\FF[x]$ be $(d,x)$-set $Q$. 
	We call $\dfield{\SA}{\sigma}$ an \rpisiE-extension of $\dfield{\FF}{\sigma}$ also \emph{$(d,x,Q)$-reduced} if
	the extension is simple and it can be rewritten (after reordering of the generators) to the form $\SA=\EE\lr{s_1}\dots\lr{s_u}$ with $\EE=\FF(x)\lr{t_1}\dots\lr{t_e}$ such that 
	\begin{itemize}[topsep=0pt, partopsep=0pt, leftmargin=7pt, itemindent=5pt]
		\item $\dfield{\EE}{\sigma}$ is an \rpisiE-extension of $\dfield{\FF(x)}{\sigma}$ with $x$-degree{$\leq d$};
		\item $\dfield{\SA}{\sigma}$ is a \sigmaE-extension of $\dfield{\EE}{\sigma}$ with $\sigma(s_i)-s_i=p_i/q_i^{m_i}$ where $m_i\geq1$, $p_i\in\FF[x]\lr{t_1}\dots\lr{t_{e}}\setminus\{0\}$ and $q_i\in Q$ with $\deg_x(p_i)<m_i\deg_x(q_i)$. The $s_1,\dots,s_u$ are also called the \emph{$Q$-contributions}.
	\end{itemize}
\end{definition}

\noindent Given such a $(d,x,Q)$-reduced \rpisiE-extension $\dfield{\SA}{\sigma}$ of $\dfield{\FF(x)}{\sigma}$ with $f\in\EE$ (by iterative application of the algorithm below), one can construct a \sigmaE-extension that is again a $(d,x,Q')$ reduced extension (with $Q\subseteq Q'$) and in which one finds $h\in\SA'$ with~\eqref{Equ:Tele}.

\begin{alg}\label{Alg:GeneralSummation}
\normalfont\small
\textbf{(Finding degree-reduced representations)}
\begin{description}[topsep=0pt, partopsep=0pt, leftmargin=3pt]
\item[\textbf{Input}]A \pisiE-field extension $\dfield{\FF(x)}{\sigma}$ of $\dfield{\FF}{\sigma}$ in which Problem~SE is solvable, $d\in\NN$, a simple \rpisiE-extension $\dfield{\EE}{\sigma}$ of $\dfield{\FF(x)}{\sigma}$ with $x$-degree{$\leq d$} which is LA-computable and in which one can solve Problem~T, and a $(d,x)$-set $Q=\{q_1,\dots,q_{\lambda}\}\subseteq\FF[x]$ such that the \sigmaE-extension $\dfield{\SA}{\sigma}$ of $\dfield{\EE}{\sigma}$ is $(d,x,Q)$-reduced \rpisiE-extension of $\dfield{\FF(x)}{\sigma}$; $f\in\EE$.
\item[\textbf{Output}] A $(d,x)$ set $Q'\supseteq Q$; a \sigmaE-extension $\dfield{\SA'}{\sigma}$ of $\dfield{\SA}{\sigma}$ which is a 
a $(d,x,Q')$-reduced \rpisiE-extension $\dfield{\SA'}{\sigma}$ of $\dfield{\FF(x)}{\sigma}$ together with a solution $h\in\SA'$ for~\eqref{Equ:Tele}.
\end{description}
\begin{enumerate}[topsep=0pt, partopsep=0pt, leftmargin=8pt,label={\footnotesize\arabic*}]
\item[] Write $\SA=\EE[s_{1,1},\dots,s_{1,e_1},\dots,s_{\lambda,1},\dots,s_{r,e_\lambda}]$ with $e_i\in\NN$, $\sigma(s_{i,j})=s_{i,j}+\frac{p_{i,j}}{q_i^{n_i}}$, $n_{i,j}\in\NP$ and $p_{i,j}\in\EE$ with $\deg_x(p_{i,j})<n_{i,j}\deg_x(q_i)$.
\item Compute $Q\subseteq\tilde{Q}=\{q_1,\dots,q_r\}\subseteq\FF[x]$ which is $(d,\den(f))$-complete; see Lemma~\ref{Lemma:UpdateCompleteSet}. Set $e_i=0$ for $\lambda<i\leq r$ (i.e., for the elements $\tilde{Q}\setminus Q$). 
\item Compute $f',g\in\EE$ such that~\eqref{Equ:RefinedTele}
where $f'$ is written in the form~\eqref{Equ:NormalForm} with the properties (1)--(4) as given in Lemma~\ref{Lemma:TransformToReducedForm}.
\item Set $Q'=\tilde{Q}$, $\SA_0=\SA$, $u=0$ and $w=0$.
\item For $i=1$ to $r$ do
\item\shift If $p_i\neq0$ then
\item\shift\shift\label{Step:ComputeBasis} If $e_i=0$ then set $B=\{\}$
\item[]\shift\shift\shift else set $\mu_i=\max(n_i,n_{i,1},\dots,n_{i,e_i})$ and compute a basis $B$ of 
%\item[]\shift\shift\shift\shift 

\vspace*{-0.35cm}

\begin{equation}\label{Equ:ViDef}
\hspace*{1cm}V_i=\Ann_{\KK}(q_i^{\mu_i-n_{i,1}}p_{i,1},\dots,q_i^{\mu_i-n_{i,{e_i}}}p_{i,e_i},q_i^{\mu_i-n_{i}}p_{i})
\end{equation}

\vspace*{-0.1cm}

\item\shift\shift\label{Step:BasisCheck} If $B=\{\}$ then
\item\shift\shift\shift\label{Step:BasisEmptyPartStart} Set $u=u+1$ and $Q'=Q'\cup\{q_i\}$.
\item\shift\shift\shift\label{Step:suExt}\myframeV{13cm}{Take a new variable $s_u$ being transcendental over $\SA_{u-1}$ and construct the difference ring extension $\dfield{\SA_u}{\sigma}$ of $\dfield{\SA_{u-1}}{\sigma}$ with $\SA_u=\SA_{u-1}[s_u]$ and $\sigma(s_u)=s_u+\frac{p_i}{q_i^{n_i}}$.}
\item\shift\shift\shift\label{Step:BasisEmptyPartEnd} Set $w=w+s_u\in\SA_u$.
\item[]\shift\shift else 
\item\shift\shift\shift\label{Step:LinSPart} Set $w=w+c_1\,s_{i,1}+\dots+c_{e_i}\,s_{i,e_i}$ where $B=\{(c_1,\dots,c_{e_i},-1)\}$.
\item[]\shift\shift fi
\item[]\shift fi
\item[] od
\item\label{Step:Compute g'} Compute, if possible, a $\gamma\in\EE$ with $\sigma(\gamma)-\gamma=\frac{p}{q}$.\\ 
If such a $\gamma$ exists, set $g'=\gamma$ and $\SA'=\SA_u$. Otherwise, define the ring extension $\dfield{\SA'}{\sigma}$ of $\dfield{\SA_u}{\sigma}$ with the polynomial ring $\SA'=\SA_u[t]$ and $\sigma(t)=t+\frac{p}{q}$, and set $g'=s$. 
\item Return $(h,\dfield{\SA'}{\sigma}, Q')$ with $h=g+g'+w\in\SA'$.
\end{enumerate}
\end{alg}

\begin{proposition}
Algorithm~\ref{Alg:GeneralSummation} is correct and can be executed in a \pisiE-field $\dfield{\FF(x)}{\sigma}$ as specified in Theorem~\ref{Thm:Computable}.
%Algorithm~\ref{Alg:GeneralSummation} satisfies the claimed specification. In particular, it can be executed if $\dfield{\FF}{\sigma}$ is a \pisiE-field over a constant field $\KK=A(y_1,\dots,y_{\lambda})$ which is a rational function field over an algebraic number field $A$. 
\end{proposition}
\begin{proof}
Consider the $i$th loop with $1\leq i\leq r$. 
For the special case $e_i=0$ in step~\ref{Step:ComputeBasis} it follows with $p_i\neq0$ that we have $V_i=\{\vect{0}\}$ with the basis $B=\{\}$. Otherwise, we compute a basis $B$ of $V_i$ and we proceed.
If $B=\{\}$ in step~\ref{Step:BasisCheck} then we adjoin a new variable $s_u$ which, for later arguments, we also denote by $s_{i,e'_{i}}$ with $e'_i=e_i+1$. 
In particular, we set $p_{i,e'_{i}}=p_i$ and $n_{i,e'_{i}}=n_i$ and get $\sigma(s_{e'_i})-s_{e'_i}=p_{i,e'_i}/{q_i}^{n_{i,e'_i}}$.\\
Otherwise, if $B\neq\{\}$, the ring will remain unchanged and we define $e'_i=e_i$. Note that $|B|=1$. Namely, suppose that we can take two elements $\vect{c}=(c_1,\dots,c_{e_i+1}), \vect{d}=(d_1,\dots,d_{e_i+1})\in B$ with $\vect{c}\neq\vect{d}$. If the last component of both vectors is nonzero, we can assume that it is $1$ by multiplying the vectors with an appropriate element of $\KK$. These normalized vectors must be still different (since $B$ is a basis). Thus $\vect{e}=\vect{c}-\vect{d}\neq\vect{0}$ where the last entry is $0$. Removing this last entry gives a vector in $\Ann_{\KK}(q_i^{\mu_i-n_{i,1}}p_{i,1},\dots,q_i^{\mu_i-n_{i,{e_i}}}p_{i,e_i})$ with $\deg_x(q_i^{\mu_i-n_{i,1}}p_{i,1})<\mu_i\,\deg_x(q_i)$. Thus we can apply Theorem~\ref{Thm:SigmaCharDeg}: $\dfield{\SA}{\sigma}$ is not a \sigmaE-extension of $\dfield{\EE}{\sigma}$, a contradiction. Hence we can suppose that $B=\{(c_1,\dots,c_{e_i},-1)\}$ as stated in step~\ref{Step:LinSPart}.\\
Now consider the difference ring extension $\dfield{\SA'}{\sigma}$ of $\dfield{\EE}{\sigma}$ of the output.
After reordering and using the renaming from above we get $\SA'=\EE'[s_{1,1},\dots,s_{1,e'_1},\dots,s_{r,1},\dots,s_{r,e'_\lambda}]$ as follows:
\begin{enumerate}[topsep=0pt, partopsep=0pt, leftmargin=7pt, itemindent=5pt]
\item $\EE'=\EE$ in case that one finds a $\gamma\in\EE$ with $\sigma(\gamma)-\gamma=\frac{p}{q}$, or $\EE'=\EE[t]$ if there is no such $\gamma$. In this case $\dfield{\EE'}{\sigma}$ with $\EE'=\EE[t]$ is a \sigmaE-extension of $\dfield{\EE}{\sigma}$ with $x$-degree{$\leq d$} by Theorem~\ref{Thm:SigmaExtPara};
\item we have $\sigma(s_{i,j})=s_{i,j}+p_{i,j}/q_i^{n_i}$ where $q_i\in Q'$, $\deg_x(p_{i,j})<n_i\,\deg_x(q_i)$ and $V_i=\{\vect{0}\}$ with~\eqref{Equ:ViDef} and with $\deg_x(q_i^{\mu_i-n_{i,1}}p_{i})<\mu_i\deg_x(q_i)$. By Theorem~\ref{Thm:SigmaCharDeg}, $\dfield{\SA'}{\sigma}$ is a \sigmaE-extension of $\dfield{\EE'}{\sigma}$.
\end{enumerate}
In summary, $\dfield{\SA'}{\sigma}$ is a $(d,x,Q')$-reduced \rpisiE-extension of $\dfield{\FF(x)}{\sigma}$. Finally, we observe that in the steps~\ref{Step:suExt} or~\ref{Step:LinSPart} we have $\sigma(s_u)-s_u=p_i/q_i^{n_i}$ with $q_i\in Q'$ or $\sigma(b)-b=p_i/q_i^{n_i}$ with $b=c_1s_{i,1}+\dots+c_{e_i}s_{i,e_i}$. Thus after quitting the for loop we get $\sigma(w)-w=\sum_{i=1}^rp_{i}/q_i^{n_i}$. With step~\ref{Step:Compute g'} and~\eqref{Equ:NormalForm}, $\sigma(g'+w)-(g'+f)=f'$. Hence with $h=g+g'+w$,
\begin{equation}\label{Equ:CombineDifference}
	\sigma(h)-h=(\sigma(g)-g)+(\sigma(g'+w)-(g'+w))\stackrel{\eqref{Equ:RefinedTele}}{=}(f-f')+f'=f.
\end{equation}
All steps are executable in $\dfield{\FF(x)}{\sigma}$ as given in Thm.~\ref{Thm:Computable}.
\end{proof}

\begin{example}
We apply Algorithm~\ref{Alg:GeneralSummation} to $\dfield{\SA}{\sigma}$ with $\SA=\EE$ given in Ex.~\ref{Exp:DR}.1, $Q=\{q_1\}=\{1+x^2\}$ and $f$ as given in~\eqref{Equ:fFromHarmonicSum}. For step~\ref{Equ:RefinedTele} see Ex.~\ref{Exp:SigmaReduceForm}.1. Since $\SA=\EE$ and $|Q|=1$, we have $r=1$ and $e_1=0$. Furthermore, $p_1=0$. Thus we enter step~\ref{Step:Compute g'} with $\frac{p}{q}=\frac{-2-4 x+x^2}{10 x^3}$.
Since there is no $\gamma\in\EE$ with $\sigma(\gamma)-\gamma=\frac{p}{q}$, we can adjoin the \sigmaE-monomial $t$ to $\EE$ with $\sigma(t)=t+\frac{p}{q}$ and get the solution
$h=g+t$ with $g$ given in~\eqref{Equ:gFromHarmonicSum}. We reinterpret $t$ as
$
\sum_{i=2}^k \frac{3-6 i+i^2}{10 (-1+i)^3}=\frac{2+4 k-k^2}{10 k^3}
-\frac{1}{5} S_3(k)
-\frac{2}{5} S_2(k)
+\frac{1}{10} S_1(k)$. Rephrasing $h$ back to the given summation objects and summing~\eqref{Equ:Tele} over $k$ from $1$ to $n$ yield the right-hand side of~\eqref{Equ:HarmonicSumId};
\end{example}

\begin{example}
	Denote the summand on the the left-hand side

\vspace*{-0.4cm}

	\begin{multline}\label{Equ:FactorialDenSum}
		\text{$\sum_{k=1}^n$} \frac{(
			1
			+k
			+(
			(1+k)^2
			+k!
			) k!
			) (1+k!)
			-k (1+k) (k!)^4 
		}{(1+k) (1
			+(1+k) k!
			) (k!)^3 (1+k!)}\text{$\sum_{i=1}^k$} \frac{1}{i!}\\[-0.3cm]
		=-\frac{1}{2}
		+\frac{1}{n!
			+n n!
		}
		-\frac{1}{1
			+n!
			+n n!
		}
		+
		\text{$\sum_{i=1}^n$} \frac{1}{(i!)^3}
		+\frac{1}{1
			+n!
			+n n!
		}\text{$\sum_{i=1}^n$} \frac{1}{i!}
	\end{multline}	
	by $F(k)$. Take the \pisiE-field $\dfield{\QQ(x)(\tau)}{\sigma}$ over $\QQ$ with $\sigma(x)=x+1$, $\sigma(\tau)=(x+1)\tau$ and consider the simple \rpisiE-extension $\dfield{\EE}{\sigma}$ of $\dfield{\QQ(x)(\tau)}{\sigma}$ with $\sigma(s)=s+\frac1{(x+1)\tau}$. Then we can rephrase $F(k+1)$ by
	$f=\frac{(
		-\tau^4 s x (1+x)
		+(1+\tau) (
		1
		+x
		+\tau (
		\tau
		+(1+x)^2
		)
		)
		)}{\tau^3 (1+\tau) (1+x) (1
		+\tau (1+x)
		)}\in\EE$. Here we set $\FF=\QQ(x)$ and the \piE-monomial $\tau$ will play the role of $x$. Note that $\tau^3$ is a period $1$ factor. Hence the extension $\dfield{\EE}{\sigma}$ of $\dfield{\QQ(x)(\tau)}{\sigma}$ has $\tau$-degree$\leq0=d$. Further, $Q=\{q_1\}$ with $q_1=\tau+1$ is $(0,f)$-complete. We apply Algorithm~\ref{Alg:GeneralSummation} with $\SA=\EE$ and get $g\in\EE$ and $f'\in\EE$ with $f'=\frac{p_1}{q_1}+\frac{p}{q}$. Namely,
	$p_1=0$, $\frac{p}{q}=\frac{1
		+x
		-\tau^2 x
	}{\tau^3 (1+x)}$ and $g=-\frac{1}{\tau}+\frac{s}{1+\tau}$. Since there is no $\gamma\in\EE'$ with $\sigma(\gamma)-\gamma=\frac{p}{q}$, we can construct the \sigmaE-extension $\dfield{\EE[t]}{\sigma}$ of $\dfield{\EE}{\sigma}$ with $\sigma(t)=t+\frac{p}{q}$. In particular, $h=g+t$ is a solution of~\eqref{Equ:Tele}. Finally, we rephrase $h$ back to summation objects. 
	Here $t$ can be interpreted as
	$\sum_{i=1}^k \frac{i^3
		+(i!)^2
		-i (i!)^2
	}{(i!)^3}=-\frac{1}{(k!)^3}
	+\frac{1}{k!}
	+
	\sum_{i=1}^k \frac{1}{(i!)^3}$. Finally, summing~\eqref{Equ:Tele} over $k$ one gets the right-hand side of~\eqref{Equ:FactorialDenSum}.	
\end{example}

\begin{example}
We want to model the sums $T_1(n)=\sum_{k=1}^n F_1(k)$ and $T_2(n)=\sum_{k=2}^n F_2(k)$ with $F_1(k)=\frac{(
	1-(-1)^j) j}{(
	3-3 j+j^2) (j!)^2}\prod_{i=1}^j i!$ and 
%\small
%\begin{align*}
%F_2(k)&=\Big(-(-1)^k (3
%+(-3+k) k
%) (1
%+(-1+k) k
%) (k!)^2\\[-0.2cm]
%&+(
%-1+(-1)^k) (-1+k) k (1+n) (1
%+(-1+k) k
%) \prod_{i=1}^k i!\\[-0.4cm]
%&+(
%1+(-1)^k) (-1+k) (3
%+(-3+k) k
%) k! \prod_{i=1}^k i!\Big)\Big/\\
%&\Big((-1+k) (
%3-3 k+k^2
%)
%(1-k+k^2) (k!)^2\Big).
%\end{align*}

\vspace*{-0.2cm}

\begin{align*}
	F_2(k)=\Big((-1)^{k-1}& (3
	+(-3+k) k
	) (1
	+(-1+k) k
	) (k!)^2+(k-1)\Big(\prod_{i=1}^k i!\Big)\times\\[-0.2cm]
	&\Big((
	-1+(-1)^k) k (1+n) (1
	+(-1+k) k
	)+(
	1+(-1)^k)  (3+(-3+k) k
	) k!\Big)\Big/\\[-0.1cm]
	&\hspace*{3.4cm}\Big((-1+k) (
	3-3 k+k^2
	)
	(1-k+k^2) (k!)^2\Big)
\end{align*}
\normalsize
in a difference ring. Here we start with $\dfield{\EE}{\sigma}$ given in Ex.~\ref{Exp:DR}.2. 
First we rephrase $F_1(k+1)$  in $\EE$ by replacing the objects $k, (-1)^k, k!, \prod_{i=1}^ki!$ with
$x,z,\tau_1,\tau_2$ yielding $f_1\in\EE$. Note that $Q=\{q_1\}$ with $q_1=x^2-x+1$ is $(1,f_1)$-complete. 
Next,
activating Algorithm~\ref{Alg:GeneralSummation} we 
get as output the $(1,x,Q)$-reduced \rpisiE-extension $\dfield{\SA}{\sigma}$ of $\dfield{\EE}{\sigma}$ with $\SA=\EE[s_{1,1}]$ and $\sigma(s_{1,1})-s_{1,1}=\frac{p_{1,1}}{q_1}(=f_1)$  where $p_{1,1}=\tau_2/\tau_1(1+z)$. Now we turn to $T_2$. As above we rephrase $F_2(k+1)$ in $\SA$ yielding $f:=f_2$. Note that $Q$ is again $(1,f)$-complete. Since there is no $h\in\EE$ with~\eqref{Equ:Tele}, we could activate Theorem~\ref{Thm:SigmaExtPara} (with $m=1$) to get the \sigmaE-monomial $t$ over $\SA$ with $\sigma(t)=t+f$. But we can do better. Activating again Algorithm~\ref{Alg:GeneralSummation} we compute 
	$g=\frac{\tau_2 (1+z)}{\tau_1 (
		1-x+x^2)}\in\EE$
	 and $f'\in\EE$ such that~\eqref{Equ:RefinedTele} holds where $f'$ has the $\sigma$-reduced form~\eqref{Equ:NormalForm} with $r=1$.	
	Namely, we get 
	$p_1=\tau_2/\tau_1 (-n
	-n z
	)$ and $\frac{p}{q}=\frac{z}{x}$. Note that this time we have $p_1\neq0$ and $e_1=1$. 
	Hence we compute for $i=1$ in step~\ref{Step:ComputeBasis}  the value $\mu_1=0$ and the basis $B=\{(-n,-1)\}$ of $\Ann_{\QQ}(p_{1,1},p_1)$ which gives $w=-n\,s(=-n\,s_{1,1})$.
	Since there is no $\gamma\in\EE$ with $\sigma(\gamma)-\gamma=\frac{p}{q}$, we can adjoin the \sigmaE-monomial $t$ to $\SA$
	 with 
	$\sigma(t)-t=\frac{p}{q}=\frac{z}{x}$ and we get the solution
	$h=w+t+g=\frac{\tau_2 (1+z)}{\tau_1 (
		1-x+x^2)}-n\,s_1+t\in\SA[t]$
	of~\eqref{Equ:Tele} as output. $t$ can be reinterpreted as $\sum_{i=2}^n(-1)^{i-1}/(i-1)$. Rephrasing $h$ with the given summation objects and summing~\eqref{Equ:Tele} over $k$ yield
	$$T_2(n)=2n 
	-n\,T_1(n)
	+\frac{1+(-1)^n}{(1-n+n^2)n!}\text{$\prod_{i=1}^n$}i!
	+
	\text{$\sum_{i=2}^n$}\frac{(-1)^{i-1}}{i-1}.$$
\end{example}

Algorithm~\ref{Alg:GeneralSummation} gives a strategy to find telescoping solutions such that the denominators have $x$-degrees{$\leq d$}. In the following we show that this is the only possible tactic and that it will always lead to a nice solution whenever it exists in some appropriate extension.

\begin{theorem}\label{Thm:TelescopingGeneralChar}
	Let $\dfield{\FF(x)}{\sigma}$ be a \pisiE-field extension of $\dfield{\FF}{\sigma}$ and let $Q=\{q_1,\dots,q_r\}\subseteq\FF[x]$ be a $(d,x)$-set. Let $\dfield{\SA}{\sigma}$ with $\SA=\EE[s_{1,1},\dots,s_{1,e_1},\dots,s_{r,1},\dots,s_{r,e_r}]$ be a $(d,x,Q)$-reduced \rpisiE-exten\-sion of $\dfield{\FF(x)}{\sigma}$ with $e_i\in\NN$ and where $s_{i,j}$ are the $Q$-contributions with $\sigma(s_{i,j})=s_{i,j}+\frac{p_{i,j}}{q_i^{n_{i,j}}}$ where 
	$n_{i,j}\in\NP$ and $p_{i,j}\in\EE$.
	For $f\in\EE$, take
	$f',g\in\EE$ such that~\eqref{Equ:RefinedTele}
	where $f'$ is written in the form~\eqref{Equ:NormalForm} with the properties (1)--(4) as given in Lemma~\ref{Lemma:TransformToReducedForm}. Then there is an $h\in\SA$ with~\eqref{Equ:Tele}	iff for all $1\leq i\leq r$ there are $c_{i,j}\in\KK=\const{\FF}{\sigma}$ with 
	\begin{equation}\label{Equ:pSigmaRel}
		q_i^{\mu_i-n_i}p_i=c_{i,1}q_i^{\mu_i-n_{i,1}}p_{i,1}+\dots+c_{i,e_i}q_i^{\mu_i-n_{i,{e_i}}}p_{i,e_i}
	\end{equation}
	for $\mu_i=\max(n_i,n_{i,1},\dots,n_{i,e_i})$ and there is a $g'\in\EE$ with 
	\begin{equation}\label{Equ:ExtractTele}
		\sigma(g')-g'=\frac{p}{q};
	\end{equation} 
	if this is the case, $\den(g')$ has $x$-degrees{$\leq d$} and we get the solution  
	\begin{equation}\label{Equ:SolForGeneralExt}
		h=g+g'+\sum_{i=1}^{r}\sum_{j=1}^{e_i}c_{i,j}s_{i,j}\in\SA.
	\end{equation}
\end{theorem} 
\begin{proof}
	Suppose there is an $h\in\SA$ with~\eqref{Equ:Tele}. Then with~\eqref{Equ:RefinedTele} we get $\sigma(\gamma)-\gamma=f'$ with $\gamma=h-g\in\SA$. In particular, by~\cite[Prop.~6.4]{AS:18} it follows that $\gamma=\sum_{i=1}^{r}\sum_{j=1}^{e_i}c_{i,j}s_{i,j}+g'$ with $c_{i,j}\in\KK$ and $g'\in\EE$. Consequently,
	$\sigma(g')-g'=\frac{p'_1}{q_1^{\mu_1}}+\dots+\frac{p'_r}{q_r^{\mu_r}}$
	with $p'_i=q_i^{\mu_i-n_i}p_i-(c_{i,1}q_i^{\mu_i-n_{i,1}}p_{i,1}+\dots+c_{i,e_i}q_i^{\mu_i-n_{i,e_i}}p_{i,e_i})$. Since $\deg_x(p'_i)<\deg_x(q_i)\mu_i$, we can apply Prop.~\ref{Prop:RemoveBadFactors} and we get $p'_1=\dots=p'_r=0$. Thus \eqref{Equ:pSigmaRel} and~\eqref{Equ:ExtractTele} hold. Furthermore, $\den(g')$ has $x$-degree {$\leq d$} by Prop.~\ref{Prop:RemoveBadFactors}. Conversely, if~\eqref{Equ:pSigmaRel} holds and there is a $g'\in\EE$ with~\eqref{Equ:ExtractTele}, then for $w=\sum_{i=1}^{r}\sum_{j=1}^{e_i}c_{i,j}s_{i,j}$ we obtain
	\begin{align*}
		\sigma(g'&+w)-(g'+w)=\frac{p}{q}+\text{$\sum_{i=1}^r$}\text{$\sum_{j=1}^{e_i}$}c_{i,j}\frac{p_{i,j}}{q_i^{n_{i,j}}}\\[-0.2cm]
		&=\frac{p}{q}+\text{$\sum_{i=1}^r$}\frac1{q_i^{\mu_i}}\text{$\sum_{j=1}^{e_i}$}c_{i,j}q_i^{\mu_i-n_{i,j}}p_{i,j}=\frac{p}{q}+\text{$\sum_{i=1}^{r}$}\frac{q_i^{\mu_i-n_i}\,p_i}{q_i^{\mu_i}}=f'.
	\end{align*} 
	Hence with~\eqref{Equ:SolForGeneralExt} it follows that~\eqref{Equ:CombineDifference} which completes the proof.
\end{proof}

\noindent With $\SA=\EE$ Theorem~\ref{Thm:TelescopingGeneralChar} reduces to Corollary~\ref{Cor:TelescopingChar}.

\begin{corollary}\label{Cor:TelescopingChar}
	Let $\dfield{\FF(x)}{\sigma}$ be a \pisiE-field extension of $\dfield{\FF}{\sigma}$ and let $\dfield{\EE}{\sigma}$ be a simple \rpisiE-extension of $\dfield{\FF(x)}{\sigma}$ with $\EE=\FF(x)\lr{t_1}\dots\lr{t_e}$ and $x$-degree{$\leq d$}. For $f\in\EE$, take
	$f',g\in\EE$ such that~\eqref{Equ:RefinedTele}
	where $f'$ is written in the form~\eqref{Equ:NormalForm} with the properties (1)--(4) as given in Lemma~\ref{Lemma:TransformToReducedForm}. Then there is an $h\in\EE$ with~\eqref{Equ:Tele} 
	if and only if $p_1=\dots=p_r=0$ and there is a $g'\in\EE$ with~\eqref{Equ:ExtractTele};
	if this is the case, $\den(g')$ has $x$-degree{$\leq d$} and $h=g+g'\in\EE$ is a solution.
\end{corollary} 

%\begin{proof}
%We apply Theorem~\ref{Thm:TelescopingGeneralChar} with $\SA=\EE$, i.e., $e_0=\dots=e_{e}$. Then the constraints~\eqref{Equ:pSigmaRel} for $1\leq i\leq r$ are equivalent to $p_1=\dots=p_r=0$ and the corollary follows.
%\end{proof}

\noindent The following ''optimal'' behavior of Algorithm~\ref{Alg:GeneralSummation} holds.

\begin{corollary}
	%Let \pisiE-extension $\dfield{\FF(x)}{\sigma}$ be a \pisiE-extension of $\dfield{\FF}{\sigma}$ in which Problem~SE is solvable and
	%let $\dfield{\EE}{\sigma}$ be a basic \rpisiE-extension  of $\dfield{\FF(x)}{\sigma}$ with $x$-degree{$\leq d$} in which one can solve Problem~T. Let  $Q\subseteq\FF[x]$ be a $(d,x)$-set such that the \sigmaE-extension $\dfield{\SA}{\sigma}$ of $\dfield{\EE}{\sigma}$ is $(d,x,Q)$-reduced \rpisiE-extension of $\dfield{\FF(x)}{\sigma}$. 
	Let $\dfield{\SA}{\sigma}$ be $(d,x,Q)$-reduced \rpisiE-extension of $\dfield{\FF(x)}{\sigma}$ with $f\in\EE$ as assumed in Algorithm~\ref{Alg:GeneralSummation} and let $Q'$, $\dfield{\SA'}{\sigma}$ and $h\in\SA'$ with~\eqref{Equ:Tele} be the output of Algorithm~\ref{Alg:GeneralSummation}. If there is a simple \rpisiE-extension $\dfield{\HH}{\sigma}$ of $\dfield{\SA}{\sigma}$ with $x$-degree{$\leq d$} with $h'\in\HH$ where $\sigma(h')-h'=f$ then the following holds.
	\begin{enumerate}[topsep=0pt, partopsep=0pt, leftmargin=7pt, itemindent=5pt,label={\arabic*.}]
		\item $Q'=Q$ and $\dfield{\SA'}{\sigma}$ is a \sigmaE-extension of $\dfield{\SA}{\sigma}$ with $x$-degree{$\leq d$}. 		
		\item If $\HH=\SA$, then $\SA'=\SA$.
		\item If $h'$ in the extension $\HH$ is free of the $Q$-contributions (i.e., free of the $s_{i,j}$) then $h$ is also free of the $Q$-contributions.
	\end{enumerate}
\end{corollary}
\begin{proof}
	Suppose that there is such an $\dfield{\HH}{\sigma}$ with $h'\in\HH$ where $\sigma(h')-h'=f$. Then $\HH=\EE'[s_{1,1},\dots,s_{1,e'_1},\dots,s_{r,1},\dots,s_{r,e'_\lambda}]$ where $\dfield{\EE'}{\sigma}$ is a simple \rpisiE-extension of $\dfield{\GG(x)}{\sigma}$ with $x$-degree{$\leq d$}. Now we apply Thm.~\ref{Thm:TelescopingGeneralChar} (with $\SA$ and $\EE$ replaced by $\HH$ and $\EE'$) and conclude that there are $c_{i,j}\in\KK$ with~\eqref{Equ:pSigmaRel} for $1\leq i\leq r$. If $e_i=0$, then $p_i=0$. Otherwise,
	we have $V_i\neq\{\vect{0}\}$ with the basis $B=\{(c_1,\dots,c_{e_i},-1)\}$. 
	Thus we do not enter in steps~\ref{Step:BasisEmptyPartStart}--\ref{Step:BasisEmptyPartEnd} and hence $Q=Q'$ and $\SA_u=\SA$.
	If $\HH=\SA$, it follows by Thm~\ref{Thm:TelescopingGeneralChar} that there is a $g'\in\SA$ with~\eqref{Equ:ExtractTele}. Thus the result is returned in $\SA'=\SA$. Otherwise, we get $\SA'=\SA_u[t]=\SA[t]$ where $t$ is a \sigmaE-monomial with $x$-degree{$\leq d$}. This proves statements~1 and~2 of the proposition. Furthermore, if $h'$ is free of the $s_{i,j}$ then it follows that $c_{i,j}=0$ in~\eqref{Equ:pSigmaRel}. In particular, $p_i=0$ for all $1\leq i\leq r$. Hence we never enter in steps~\ref{Step:ComputeBasis}--\ref{Step:LinSPart} and thus $w=0$. Therefore $h$ is free of the $s_{i,j}$ and statement 3 is proven.
\end{proof}

%\begin{corollary}
%	Let $\dfield{\FF(x)}{\sigma}$ be a \pisiE-field extension of $\dfield{\FF}{\sigma}$ and let $\dfield{\EE}{\sigma}$ be a simple \rpisiE-extension of $\dfield{\FF(x)}{\sigma}$ with $x$-degree{$\leq d$}; 	
%	let  $f\in\EE$. Take the multiplicative group $G=[\FF(x)^*]_{\FF(x)}^{\AR}$, and let $\dfield{\EE'}{\sigma}$ and $\dfield{\HH}{\sigma}$ be both $G$-simple \rpisiE-extensions of $\dfield{\EE}{\sigma}$ with $x$-degree{$\leq d$}. If there is an $h\in\HH$ with $\sigma(h)-h=f$ then there is a \sigmaE-extension $\dfield{\SA}{\sigma}$ of $\dfield{\EE'}{\sigma}$ with $x$-degree {$\leq d$} in which we find $h'\in\SA$ with $\sigma(h')-h'=f$.
%\end{corollary}
%\begin{proof}
%	Note that $\dfield{\EE}{\sigma}$ and $\dfield{\HH}{\sigma}$ are simple \rpisiE-extensions of $\dfield{\FF(x)}{\sigma}$.
%	For $f\in\EE$ take
%	$f',g\in\EE$ such that~\eqref{Equ:RefinedTele}
%	where $f'$ is written in the form~\eqref{Equ:NormalForm} with the properties (1)--(4) as given in Lemma~\ref{Lemma:TransformToReducedForm}. By Corollary~\ref{Cor:TelescopingChar} and the existence of $h\in\HH$ with $\sigma(h)-h=f$ it follows that $p_1=\dots=p_r=0$. Suppose that there is a $g'\in\EE$ with $\sigma(g')-g'=\frac{p}{q}$. Then it follows by again by Corollary~\ref{Cor:TelescopingChar} that $\sigma(h')-h'=f$ with $h'=g'+g\in\SA:=\EE$. Otherwise we can take the \sigmaE-extension $\dfield{\SA}{\sigma}$ of $\dfield{\EE}{\sigma}$ with $\SA=\EE[s]$ with $\sigma(s)=s+\frac{p}{q}$ and $x$-degree{$\leq d$} and get again the desired solution $h'=s+g\in\SA$. 
%\end{proof}

The above results can be turned to parameterized versions. Here we extend only Corollary~\ref{Cor:TelescopingChar} yielding Algorithm~\ref{Alg:GeneralSummation} below.

\begin{corollary}\label{Cor:ParaT}
Let $\dfield{\FF(x)}{\sigma}$ be a \pisiE-field extension of $\dfield{\FF}{\sigma}$ and let $\dfield{\EE}{\sigma}$ be a simple \rpisiE-extension of $\dfield{\FF(x)}{\sigma}$ with $\EE=\FF(x)\lr{t_1}\dots\lr{t_e}$ and $x$-degree{$\leq d$}. For $f_1,\dots,f_m\in\EE^*$ let $Q=\{q_1,\dots,q_r\}\subseteq\FF[x]$ be $(d,f_i)$-complete with $1\leq i\leq m$. Take
$f'_i,g_i\in\EE$ with $\sigma(g_i)-g_i+f'_i=f_i$
where $f'_i$ is given by

\vspace*{-0.2cm}

\begin{equation}\label{Equ:NormalFormi}
f'_i=\frac{p_{i,1}}{q_1^{n_{i,1}}}+\frac{p_{i,r}}{q_r^{n_{i,1}}}+\frac{p'_i}{q'_i}
\end{equation}

\vspace*{-0.0cm}

\noindent with the properties (1)--(4) ($p_j,n_j$ replaced by $p_{i,j},n_{i,j}$, and $p,q$ replaced by $p'_i,q'_i$) as given in Lemma~\ref{Lemma:TransformToReducedForm}. Let $\mu_i=\max(n_{1,i},\dots,n_{m,i})$ and

\vspace*{-0.5cm}

\begin{equation}\label{Equ:ParaVi}
V_i=\Ann_{\KK}(q_1^{\mu_1-n_{1,j}}p_{1,j},\dots,q_1^{\mu_m-n_{m,j}}p_{m,j}).
\end{equation}
Then with $\vect{c}=(c_1,\dots,c_m)\in\KK^m$ the following holds.
\begin{enumerate}[topsep=0pt, partopsep=0pt, leftmargin=7pt, itemindent=5pt,label={\arabic*.}]
\item There is an $h\in\EE$  with~\eqref{Equ:ParaT} iff there is a $g'\in\EE$ with

\vspace*{-0.2cm}
\begin{equation}\label{Equ:ParafP}
\sigma(g')-g'=c_1\,\frac{p'_1}{q'_1}+\dots+c_m\,\frac{p'_m}{q'_m}
\end{equation}
and $\vect{c}\in V_i$ holds for all $1\leq j\leq r$;
if this is the case, $\den(g')$ has $x$-degree{$\leq d$} and $h=g'+\sum_{i=1}^mc_i\,g_i\in\EE$ is a solution of~\eqref{Equ:ParaT}. 
\item If $\vect{c}\in V_i$ for all $1\leq j\leq r$ but there is no $g'\in\EE$ with~\eqref{Equ:ParafP},
then there is the \sigmaE-extension $\dfield{\EE[t]}{\sigma}$ of $\dfield{\EE}{\sigma}$ with $x$-degree{$\leq d$} where $\sigma(t)-t=c_1\,\frac{p'_1}{q'_1}+\dots+c_m\,\frac{p'_m}{q'_m}$, and $h=t+\sum_{i=1}^mc_i\,g_i$ satisfies~\eqref{Equ:ParaT}. 
\end{enumerate}
\end{corollary}

\begin{proof}
(1) Define $f:=\sum_{i=1}^mc_i\,f_i$, $g:=\sum_{i=1}^mc_i\,g_i$, $f'=\sum_{i=1}^mc_i\,f'_i$, $p_j=\sum_{i=1}^mc_i\,p_{i,j}$, 
and $p\in\FF[x]\lr{t_1}\dots\lr{t_e}$, $q\in\GG[x]$ such that $\frac{p}{q}:=\sum_{i=1}^mc_i\frac{p'_i}{q'_i}$.
With $f_i=\sigma(g_i)-g_i+f'_i$ for $1\leq i\leq m$ it follows

\vspace*{-0.3cm}

\begin{equation}\label{Equ:ParaCombined}
f=\text{$\sum_{i=1}^m$}c_i f_i=\sigma\Big(\text{$\sum_{i=1}^m$}c_ig_i\Big)-\text{$\sum_{i=1}^m$}c_ig_i+\text{$\sum_{i=1}^m$}c_if'_i=\sigma(g)-g-f'.
\end{equation} 
One can verify that for $f'$ the representation~\eqref{Equ:NormalForm} with the properties (1)--(4) hold. Thus by Corollary~\ref{Cor:TelescopingChar} if follows that there is an $h\in\EE$ with~\eqref{Equ:Tele} (i.e., \eqref{Equ:ParaT} holds) if and only if $p_j=0$ for all $1\leq j\leq r$ (i.e., $\vect{c}\in V_j$ for all $j$) and there is a $g'\in\EE$ with~\eqref{Equ:ExtractTele} (i.e. \eqref{Equ:ParafP} holds). Finally, 
the irreducible factors in the denominator of $g'$ have $x$-degrees{$\leq d$} and $h=g'+g\in\EE$ is a solution of~\eqref{Equ:ParaT}.\\
(2) Suppose that $\vect{c}\in V_i$ for all $1\leq j\leq r$. Since there is no $g'\in\EE$ with~\eqref{Equ:ParafP}, we can apply Theorem~\ref{Thm:SigmaExtPara} with $m=1$, and it follows that $t$ as given in statement~2 is a \sigmaE-monomial over $\EE$ where $\den(\sigma(t)-t)$ has $x$-degree{$\leq d$}. 
Furthermore, $g'=t$ is a solution of~\eqref{Equ:ParaT}. Hence we can apply statement~1 by replacing $\EE$ with $\EE[t]$ and it follows that $h=g'+g=t+g$ with $c_1,\dots,c_m$ is a solution of~\eqref{Equ:ParaT}.
\end{proof}

%\noindent This yields the following parameterized version of Algorithm~\ref{Alg:GeneralSummation}.

\begin{alg}\label{Alg:ParaTele}
	\normalfont \small
	\textbf{(Refined parameterized telescoping)}
\begin{description}[topsep=0pt, partopsep=0pt, leftmargin=3pt]
	\item[\textbf{Input:}]A \pisiE-field extension $\dfield{\FF(x)}{\sigma}$ of $\dfield{\FF}{\sigma}$ in which Problem~SE is solvable; $d\in\NN$  and an LU-computable simple \rpisiE-extension $\dfield{\EE}{\sigma}$ of $\dfield{\FF(x)}{\sigma}$ with $x$-degree{$\leq d$} in which Problem~PT is solvable; $\vect{f}=(f_1,\dots,f_m)\in\EE^m$.
	% where for all $1\leq i\leq m$ there is an irreducible period $0$ irreducible factor of $\den(f_i)$ which has degree {$>d$}.
\item[\textbf{Output:}] A solution $\vect{c}=(c_1,\dots,c_m)\in\KK^m\setminus\{\vect{0}\}$ and $h\in\EE$ of~\eqref{Equ:ParaT} if it exists. Otherwise a constructive decision if there is a \sigmaE-extension $\dfield{\SA}{\sigma}$ of $\dfield{\EE}{\sigma}$ with $x$-degree{$\leq d$} with a solution of 
$\vect{c}\in\KK^m\setminus\{\vect{0}\}$ and $h\in\SA$ of~\eqref{Equ:ParaT}.  
\end{description}
\begin{enumerate}[topsep=0pt, partopsep=0pt, leftmargin=8pt,label={\footnotesize\arabic*}]
	\item\label{stepPT:Start} For $f_1,\dots,f_m\in\EE$, compute $Q=\{q_1,\dots,q_r\}\subseteq\FF[x]$ which is $(d,f_i)$-complete; here one may use a variant of Lemma~\ref{Lemma:UpdateCompleteSet}. Further,
	compute
	$f'_i,g_i\in\EE$ such that $\sigma(g_i)-g_i+f'_i=f_i$
	where $f'_i$ is written in the form~\eqref{Equ:NormalFormi}
	with the properties (1)--(4) ($p_j,n_j$ replaced by $p_{i,j},n_{i,j}$, and $p,q$ replaced by $p'_i,q'_i$) as given in Lemma~\ref{Lemma:TransformToReducedForm}. 
	\item\label{stepPT:stop1} Compute for $1\leq i\leq r$ the bases $B_i$ of $V_i$ given in~\eqref{Equ:ParaVi} with $\mu_i=\max(n_{1,i},\dots,n_{m,i})$. 
	If $B_i=\{\}$, then stop and output ``no solution''. 
	\item\label{stepPT:stop2} Compute a basis $B=\{(d_{i,1},\dots,c_{i,m})\}_{1\leq i\leq u}$ of $V=V_1\cap\dots\cap V_m$.\\ 
	If $B=\{\}$, i.e., $V=\{\vect{0}\}$ then stop with the output ``no solution''.  
	\item Compute $(\tilde{f}_1,\dots,\tilde{f}_u)^t=D\,\Big(\frac{p'_1}{q'_1},\dots,\frac{p'_m}{q'_m}\Big)^t$ with $D=(d_{i,j})\in\KK^{u\times m}$.
	\item\label{stepPT:solinE} Compute, if possible, $\vect{\kappa}=(\kappa_1,\dots,\kappa_u)\in\KK^u\setminus\{\vect{0}\}$ and $\tilde{g}\in\EE$ with 

\vspace*{-0.2cm}

	\begin{equation}\label{Equ:SimpleCrea}
	\sigma(\tilde{g})-\tilde{g}=\kappa_1\,\tilde{f}_1+\dots+\kappa_u\,\tilde{f}_u.
	\end{equation}
	\item\label{stepPT:solinS} If there is not such a solution, then take the difference ring extension $\dfield{\SA}{\sigma}$ of $\dfield{\EE}{\sigma}$ with the polynomial ring extension $\SA=\EE[t]$ and $\sigma(t)=t+\tilde{f}_1$. Set $\tilde{g}=t$, $\vect{\kappa}=(1,0,\dots,0)\in\KK^u$.
	\item\label{StepPTFinalOutput}  Return $(\vect{c},h,\dfield{\SA}{\sigma})$ with $\vect{c}=\vect{\kappa}D\in\KK^m$, $h=\tilde{g}+\sum_{i=1}^mc_i\,g_i\in\SA$.		
	\end{enumerate}
\end{alg}

\begin{proposition}
Algorithm~\ref{Alg:ParaTele} is correct and can be executed in a \pisiE-field $\dfield{\FF(x)}{\sigma}$ as specified in Theorem~\ref{Thm:Computable}.
\end{proposition}
\begin{proof}
Suppose that $V=V_1\cap\dots\cap V_m=\{\vect{0}\}$. By Cor.~\ref{Cor:ParaT}.1 there are no $\vect{c}=(c_1,\dots,c_m)\in\KK^m\setminus\{\vect{0}\}$ and $h$ in $\dfield{\EE}{\sigma}$ or in a \sigmaE-extension $\dfield{\SA}{\sigma}$ of $\dfield{\EE}{\sigma}$ with $x$-degree{$\leq d$} which satisfy~\eqref{Equ:ParaT}. Thus the stops in steps~\ref{stepPT:stop1} and~\ref{stepPT:stop2} with the output ''no solution'' are correct. Now suppose that there is such a solution $h\in\EE$ with $\vect{c}\neq\vect{0}$. We conclude with Corollary~\ref{Cor:ParaT}.1 that $\vect{c}\in V$ holds and that there is a $g'\in\EE$ with~\eqref{Equ:ParafP}. Since $B$ is a basis of $V$, there is a $\vect{b}=(b_1,\dots,b_u)\in\KK^u\setminus\{\vect{0}\}$ with $\vect{b}\,D=\vect{c}\neq\vect{0}$. Consequently 
$$b_1\,\tilde{f}_1+\dots+b_u\,\tilde{f}_u=\vect{b}(\tilde{f}_1,\dots,\tilde{f}_u)^t=\overbrace{\vect{b}\,D}^{\vect{c}}(\frac{p'_1}{q'_1},\dots,\frac{p'_m}{q'_m})^t=\sigma(g')-g'.$$
Thus we also find $\vect{\kappa}=(\kappa_1,\dots,\kappa_u)\in\KK^u\setminus\{\vect{0}\}$ and $\tilde{g}\in\EE$ with~\eqref{Equ:SimpleCrea} in step~\ref{stepPT:solinE}. Now consider the output given in step~\ref{StepPTFinalOutput}. Then
\begin{align*}
	c_1\,f'_1&+\dots+c_m\,f'_m=\vect{c}(f'_1,\dots,f'_m)^t=\vect{\kappa}D(f'_1,\dots,f'_m)^t\\[-0.1cm]
	\stackrel{(*)}{=}&\vect{\kappa}D(\frac{p'_1}{q'_1},\dots,\frac{p'_m}{q'_m})^t=\vect{\kappa}(\tilde{f}_1,\dots,\tilde{f}_u)^t=\sigma(\tilde{g}))-\tilde{g};
\end{align*}
here (*) holds since $D$ kills the contributions with denominator factors having $x$-degrees larger than $d$.  With $\sigma(g_i)-g_i+f'_i=f_i$ and~\eqref{Equ:ParaCombined} we obtain $\sum_{i=1}^mc_i\,f_i=\sigma(h)-h$ with $h=g'+\sum_{i=1}^mc_ig_i$. Summarizing, if once can solve Problem~PT in $\EE$, the algorithm finds such a solution. Otherwise, we fail to find $\kappa_i\in\KK$ and $\tilde{g}\in\EE$ with~\eqref{Equ:SimpleCrea}. By Thm.~\ref{Thm:SigmaExtPara}, $\dfield{\SA}{\sigma}$ given in step~\ref{stepPT:solinS} is a \sigmaE-extension of $\dfield{\EE}{\sigma}$ with $x$-degree{$\leq d$}; Further, for the $\kappa_i$ and $g'=t$ we have~\eqref{Equ:SimpleCrea}. Also for this case the output $h\in\SA$ produces the desired solution.
Clearly, the algorithm is applicable as specified in Theorem~\ref{Thm:Computable}.
\end{proof}

\vspace*{-0.3cm}

\begin{example}
Given $\dfield{\EE}{\sigma}$ from Ex.~\ref{Exp:DR}.3 and $\vect{f}\in\EE$ as given in~\eqref{Exp:VectorF} we start Algorithm~\ref{Alg:ParaTele}. Step~\ref{stepPT:Start} has been carried out in Ex.~\ref{Exp:SigmaReduceForm}.2.
Next, we compute a basis of $\Ann_{\QQ}(p_{1,1},p_{2,1},p_{3,1})$. Namely, $B=\{\vect{d}\}$ with $\vect{d}=(d_1,d_2,d_3)=(1,-1,1)$.
Finally, we get $\tilde{f}_1=\sum_{i=1}^3d_i\rho_i=\frac{h_1 z}{x}+\frac{z}{x}$ with $u=1$. Since there is no $\kappa_1\neq0$ and $\tilde{g}\in\EE$ with~\eqref{Equ:SimpleCrea}, we can adjoin the \sigmaE-monomial $t$ with $\sigma(t)-t=\tilde{f}_1=\frac{h_1 z}{x}+\frac{z}{x}$  and set $\vect{\kappa}=(1)$ and $\tilde{g}=t$. 
Thus $\vect{c}=\vect{\kappa}(d_1,d_2,d_3)=(1,-1,1)$ and $h=\tilde{g}+g_1-g_2+g_3=t
	-\frac{\text{h1} x
		-z
	}{x (
		1+x^2)}$
	is a solution of~\eqref{Equ:ParaT} with $m=3$.
	%\begin{align*}
	%h(k)&=-\frac{-(-1)^k
		%	+k S_{-1}({k})
		%}{k (
		%	1+k^2)}
	%+
	%\sum_{i=2}^k -\frac{(-1)^i (
		%	-(-1)^i
		%	+i
		%	+i S_{-1}({i})
		%	)}{(-1+i) i}\\
	%&\frac{1}{2}
	%-\frac{k (-1)^k}{1+k^2}
	%+\frac{k^2 S_{-1}({k})}{1+k^2}
	%-\frac{(-1)^k S_{-1}({k})}{k}
	%+\frac{1}{2} S_{-1}({k})^2
	%+\frac{1}{2} 
	%\sum_{i_1=2}^k \frac{1}{i_1^2}
	%\end{align*}
\end{example}

\noindent\textit{Remark.} If one is interested in solving Problem PT only in $\EE$ (and not in an extension $\SA$), one could also use the algorithms from~\cite{DR1,DR3}; compare Thm.~\ref{Thm:Computable}. However, similarly to the observation in~\cite{ChenHKL15}, we have the benefit that Algorithm~\ref{Alg:ParaTele} (in comparison to the ones in~\cite{DR1,DR3}) leads to speedups when complicated denominators arise.

\section{Conclusion}\label{Sec:Conclusion}

We presented telescoping algorithms that enable one to decide algorithmically if irreducible factors can be eliminated in the input summand $f$. 
The algorithms require that the nested sums arising in $f$ have already representations with nice denominators. In order to be more flexible, we considered sum extensions in Def.~\ref{Def:reducedExt} where the outermost sum can have bad denominators. It would be interesting to see if this can be pushed further to more complex sums. Here ideas from~\cite{Abramov:75,Paule:95,Abramov:03,BChen:05,Schneider:07d,CSFFL:15} 
might be useful to eliminate denominator factors within unwanted shift-equivalence classes. 

Further, one could try to deal with several \pisiE-field monomials (and not only $x$). This would lead to algorithms that can handle not only the ($q$--)rational but also the multibasic and mixed case~\cite{Bauer:99}.

The above algorithm have been implemented in the package \texttt{Sigma} and are combined with algorithms given in~\cite{Schneider:08c,Schneider:10a,Schneider:10b,Petkov:10,Schneider:15} when one has to solve Problems~T and~PT in 
step~\ref{Step:Compute g'} of Alg.~\ref{Alg:GeneralSummation} and step~\ref{stepPT:solinE} of Alg.~\ref{Alg:ParaTele}. Thus one can search in addition for sum representations with optimal nesting depth. This highly flexible toolbox is crucial to simplify complicated sum expressions coming, e.g., from particle physics~\cite{BKKS:09,BMSS:22a,BMSS:22b}.

%In~\cite{Schneider:10c,DR3} (see also~\cite{Singer:08}) it has been shown that the non-existence of a solution of~\eqref{Equ:ParaT} proves the algebraic independence of the sums defined over the $f_i$. Utilizing, e.g., Theorem~\ref{Cor:ParaT} could provide new transcendence proofs for big classes of indefinite nested sums.\\

%\vspace*{1cm}
%\vspace{5mm}\noindent
%{\bf Acknowledgments.}~~This project has received funding from the Austrian Science Fund (FWF) 
%grant P33530.

%\bibliographystyle{abbrv}
%\bibliography{biblio}

\end{document}